\documentclass[a4paper,11pt,oneside]{article}
\usepackage[margin=2cm]{geometry}
\usepackage{amsmath, amsthm, amssymb}
\usepackage{hyperref}
\usepackage{mathtools}
\usepackage{natbib}
\usepackage{centernot}

\hypersetup{
	colorlinks=true,
  linkcolor=red,          %
  citecolor=blue,         %
  filecolor=magenta,      %
  urlcolor=cyan           %
}

\usepackage[utf8]{inputenc} %
\usepackage[T1]{fontenc}    %
\usepackage{hyperref}       %
\usepackage{url}            %
\usepackage{booktabs}       %
\usepackage{amsfonts}       %
\usepackage{nicefrac}       %
\usepackage{microtype}      %

\usepackage{mathtools}
\usepackage{amsthm, amssymb}

\newcommand{\bepf}{\begin{proof}}
\newcommand{\enpf}{\end{proof}}

\newtheorem{theorem}{Theorem}
\newtheorem{lemma}{Lemma}

\newtheorem{proposition}{Proposition}

\newtheorem{conjecture}{Conjecture}
\newtheorem{corollary}{Corollary}

\newcommand{\set}[1]{\left\{ #1 \right\}}

\newcommand{\pred}[1]{\boldsymbol{1}[#1]}

\newcommand{\R}{\mathbb{R}}

\newcommand{\N}{\mathbb{N}}

\newcommand{\abs}[1]{\left| #1 \right|}
\newcommand{\paren}[1]{\left( #1 \right)}
\newcommand{\sqprn}[1]{\left[ #1 \right]}
\newcommand{\nrm}[1]{\left\Vert #1 \right\Vert}
\newcommand{\norm}[1]{\Vert #1 \Vert}

\newcommand{\vertiii}[1]{{\left\vert\kern-0.25ex\left\vert\kern-0.25ex\left\vert #1 
    \right\vert\kern-0.25ex\right\vert\kern-0.25ex\right\vert}}
\newcommand{\tv}[1]{\nrm{#1}_{\textup{\tiny\textsf{TV}}}}
\newcommand{\ello}[1]{\nrm{#1}_{1}}

\DeclarePairedDelimiter\ceil{\lceil}{\rceil}
\DeclarePairedDelimiter\floor{\lfloor}{\rfloor}

\newcommand{\mexp}{\mathbb{E}}
\newcommand{\X}{\boldsymbol{X}}
\newcommand{\PR}[2][]{\mathop{\mathbb{P}}_{#1}\left( #2 \right)}
\newcommand{\E}[2][]{\mathop{\mexp}_{#1}\left[ #2 \right]}
\newcommand{\Var}[2][]{\mathbb{V}\mathbf{ar}_{#1}\left[ #2 \right]}
\newcommand{\dist}{\bmu}

\newcommand{\empdist}{\hat\dist_n}

\newcommand{\eps}{\varepsilon}

\DeclareMathOperator{\Unif}{Uniform}

\newcommand{\risk}{\mathcal{R}}

\newcommand{\bmu}{\boldsymbol{\mu}}
\newcommand{\bnu}{\boldsymbol{\nu}}

\newcommand{\bxi}{{\xi}}

\newcommand{\bigO}{\mathcal{O}}

\newcommand{\kl}[2]{D_{
\mathrm{KL}
  }
  \left(#1 || #2\right)
  }
\newcommand{\eqdef}{:=}

\usepackage{xcolor}

\newcommand{\decr}[1]{{#1}^{\downarrow}}

\newcommand{\beq}{\begin{eqnarray*}}
\newcommand{\eeq}{\end{eqnarray*}}
\newcommand{\beqn}{\begin{eqnarray}}
\newcommand{\eeqn}{\end{eqnarray}}

\newcommand{\QED}{\hfill\ensuremath{\square}}

\usepackage{ifmtarg}
\usepackage{xifthen}%
\newcommand{\ent}[1][]{%
\ifthenelse{\isempty{#1}}{%
\mathrm{H}
}{
\mathrm{H}^{(#1)}
}}

\newcommand{\loch}[1][]{%
\ifthenelse{\isempty{#1}}{%
\mathrm{h}
}{
\mathrm{h}^{(#1)}
}}

\newcommand{\mathe}{\mathrm{e}}

\def\longto{\mathop{\longrightarrow}\limits}
\newcommand{\ninf}{\longto_{n\to\infty}}

\usepackage[ruled,vlined]{algorithm2e}

\usepackage{import}

\usepackage{amsthm}
\usepackage{mathtools}
\usepackage{bbm}

\title{Dimension-Free Empirical Entropy Estimation}

\author{%
  Doron Cohen  \\
  Department of Computer Science \\
	Ben-Gurion University of the Negev \\
  Beer-Sheva, Israel \\
  \texttt{doronv@post.bgu.ac.il} \\
  \and
  Aryeh Kontorovich \\
  Department of Computer Science \\
	Ben-Gurion University of the Negev \\
  Beer-Sheva, Israel \\
  \texttt{karyeh@cs.bgu.ac.il} \\
  \and
    Aaron Koolyk \\
  Department of Computer Science \\
	Hebrew University, 
  Jerusalem, Israel \\
  \texttt{aaron.koolyk@mail.huji.ac.il} \\
  \and
  Geoffrey Wolfer \\
  RIKEN, Center for AI Project \\
  Tokyo, Japan \footnote{Part of this research was conducted when the author was graduate student at Ben-Gurion University of The Negev, Israel. } \\
	\texttt{geo-wolfer@m2.tuat.ac.jp} \\
}

\newcommand{\mathd}{\mathrm{d}}

\begin{document}

\maketitle

\begin{abstract}
  We seek an entropy estimator for discrete distributions
  with fully empirical accuracy bounds. 
  As stated, this
  goal is infeasible without some prior assumptions on
  the distribution. We discover that a certain information moment assumption
  renders the problem feasible. We argue that the moment assumption
  is natural and, in some sense, {\em minimalistic}
  --- weaker than finite support or tail decay conditions.
  Under the moment assumption,
  we provide the first
  finite-sample entropy estimates
  for infinite alphabets,
nearly recovering the known minimax rates.
Moreover, we demonstrate that our empirical bounds 
  are
  significantly
  sharper
  than 
the state-of-the-art bounds, for various natural distributions and non-trivial sample regimes.
  Along the way, we give a dimension-free analogue of the Cover-Thomas
  result on entropy continuity (with respect to total variation distance)
  for finite alphabets, which may be of independent interest.
 Additionally, we resolve 
 all of the
 open problems
 posed by
 J\"urgensen and Matthews, 2010.
\end{abstract}

\section{Introduction}
\label{sec:intro}

Estimating the entropy of a discrete distribution based on a finite iid sample is a classic problem with theoretical and practical ramifications. Considerable progress has been made
in the case of a finite alphabet, and the countably infinite
case has also attracted a fair amount of attention in
recent years. 
(See Section~\ref{sec:rel-work}
for a some background, motivation, and related work.)
A less-addressed issue is one of {\em empirical}
accuracy estimates: data-dependent bounds 
adaptive
to
the particular distribution being sampled.

Our point of departure is
the simpler 
(to analyze)
problem of estimating a discrete distribution $\dist$ in total variation norm $\tv{\cdot}=\frac12\nrm{\cdot}_1$,
where
the most recent advance
was made by \citet{CKW20}; see therein for a literature review. If $\bmu$ is a distribution on $\N$ and $\empdist$ is its empirical realization based on a sample of size $n$, then Theorem 2.1 of \citeauthor{CKW20} states that
with probability at least $1-\delta$,
\beqn
\label{eq:ckw20}
\ello{\bmu-\empdist} &\le& \frac2{\sqrt n}\sum_{j\in\N}\sqrt{\empdist(j)}
+6\sqrt{\frac{\log(2/\delta)}{2n}}.
\eeqn
This bound has the advantage of being valid for all distributions on $\N$,
without any prior assumptions, and being fully empirical: it yields a
risk 
estimate
that is computable based on the observed sample, not depending on
any unknown quantities.
(Additionally, \citeauthor{CKW20}
argue that (\ref{eq:ckw20})
is near-optimal in a well-defined sense.)
The question we set out to explore in this paper is: What analogues of
\eqref{eq:ckw20}
are possible for discrete entropy estimation?

When $\bmu$ has support size $d<\infty$, an answer to our question is
readily provided by combining \eqref{eq:ckw20} with \citet[Theorem 17.3.3]{cover2001elements}, which asserts that,
for
$\ello{\bmu-\bnu}\le1/2$,
we have
\beqn
\label{eq:CT}
\abs{\ent(\bmu)-\ent(\bnu)}
&\le& 
\ello{\bmu-\bnu}\log\frac{d}{\ello{\bmu-\bnu}}
,
\eeqn
where $\ent(\cdot)$ is the entropy functional defined in
\eqref{eq:halpha}.
Indeed, taking $\bmu$ as in \eqref{eq:ckw20}
and $\bnu$ to be $\empdist$ yields a fully empirical estimate on
$\abs{\ent(\bmu)-\ent(\empdist)}$.
For fixed $d<\infty$,
no technique relying
on the plug-in estimator
can yield  minimax rates
\citep{wu2016minimax}.
The plug-in is, however, 
asymptotically optimal
for fixed $d<\infty$
\citep{paninski2003estimation}
as well as strongly universally consistent even for $d=\infty$
\citep{antos2001convergence},
and is 
among the few methods
for which
explicitly computable finite-sample
risk bounds are known.

The thrust of this paper is to replace the restrictive
finite-support assumption with 
considerably
more general
moment conditions. 
It is well-known that
when estimating the mean of some random variable
$X
$, the first-moment assumption
$\mexp|X|\le M$ is not sufficient to yield any
finite-sample information.\footnote{
Even distinguishing,
for $X\ge0$,
between $\mexp X=0$
and $\mexp X=M$ 
based on a finite sample
is impossible with any degree of confidence.
Of course,
$\frac1n\sum_{i=1}^n X_i\to\mexp X$
almost surely, by the strong
law of large numbers.
}
Strengthening the assumption to
$\mexp|X|^\alpha\le M$, for any $\alpha>1$,
immediately yields
finite-sample empirical estimates
on $\abs{\mexp X-\frac1n\sum_{i=1}^n X_i}$
via
the
\citet{bahr-esseen-ineq}
inequality.\footnote{
Put $Y=X-\mexp X$;
then $\mexp|Y|\le2M$.
For
$1<\alpha<2$,
a sharper version of the 
Bahr-Esseen inequality
\citep{pinelis-best-15}
states that 
$\E{\abs{\sum_{i=1}^n Y_i}^\alpha}\le 2n(2M)^\alpha$,
which implies tail bounds via
Markov's inequality.
Better rates are available via the
median-of-means estimator, see \citet{DBLP:journals/focm/LugosiM19}.
}
In this sense, a bound on the $(1+\eps)$th moment
is a {\em minimal} requirement
for empirical mean estimation.
However, it is not immediately obvious how to apply this
insight to the entropy estimation problem:
the corresponding random variable is $X=-\log\bmu(I)$,
where $I\sim \bmu$, but rather than being given
iid samples of $X$, we are only given draws of $I$.
In Corollary~\ref{cor:adaptive-hybrid-bound},
we provide an empirical
entropy estimate under
a
$(1+\eps)$th moment
assumption
(for any $\eps>0$)
on
$X=-\log\bmu(I)$.

\paragraph{Our contribution.}
In Theorem~\ref{thm:dimfree},
we provide a dimension-free analogue of \eqref{eq:CT}, which,
combined with \eqref{eq:ckw20}, allows for empirical accuracy bounds
on the plug-in entropy estimator under a minimalistic moment assumption.
Moreover, for this 
rich
class
of distributions, the plug-in estimator turns
out to be
asymptotically
optimal, as we show in 
Theorem~\ref{thm:plug-in-minimax}.
Our moment assumption is natural and 
considerably less restrictive than the finite-alphabet %
and tail conditions studied in previous works
(see 
Sections
\ref{sec:moments}
and
\ref{sec:moments-tails}).
Moreover,
as we argue
in Theorem~\ref{thm:no-emp},
without such a moment assumption,
an empirical bound is not feasible.
As we demonstrate in Section~\ref{sec:rates},
the rates provided by our empirical bound
compare favorably against the state of the art.

\section{Definitions and notation}
\label{sec:defs}

Our logarithms will always be base $\mathe$ by default.
For discrete distributions, there is no loss of generality in
taking
the domain
to be
the natural numbers
$\N=\set{1, 2, 3, \dots}$.
For $k\in\N$, we write $[k]\eqdef\set{
i\in\N:i\le k
}$.
The set of all probability distributions
on $\N$
will be denoted by $\Delta_\N$.
For $d\in\N$, we write $\Delta_d\subset\Delta_\N$
to denote those $\dist$ whose support is contained in $[d]$.

We 
define
the operator
$\decr{(\cdot)}$,
which maps any $\bmu\in\Delta_\N$
to its non-increasing rearrangement
$\decr{\bmu}$.
The set of all
non-increasing distributions
will be denoted by
$\Delta_\N^{\downarrow}
:=
\set{\decr{\bmu}:\bmu\in\Delta_\N}
$.

We write $\R_+:=[0,\infty)$.
For any $\bxi:\N\to\R_+$ and 
$\alpha\ge0$, define
\beqn
\label{eq:halpha}
\ent[\alpha](\bxi) := \sum_{j\in\N
:\bxi(j)>0
}
\xi(j)
\abs{\log
  \bxi(j)
  }^\alpha
  .
\eeqn
For $\bxi\in\R^\N$, denote by $\abs{\bxi}\in\R_+^\N$ the elementwise
application of $\abs{\cdot}$ to $\bxi$.
When
$\bxi
\in\Delta_\N
$
and
$\alpha=1$,
\eqref{eq:halpha}
recovers the standard definition of entropy,
which we denote by $
\ent(\bxi) :=
\ent[1](\bxi)
$. For general
$\alpha>0$, this quantity may be referred to as the $\alpha$th {\em moment
  of information}.
For $h\ge0$, define
\beq
\Delta_\N^{(\alpha)}[h]
=\set{\bmu\in
\Delta_\N
:
\ent[\alpha](\bmu)\le h}
\eeq
and also
$
\Delta_\N^{(\alpha)}
:=
\bigcup_{h\ge0}\Delta_\N^{(\alpha)}[h]
$
and
$
\Delta_\N^{\downarrow(\alpha)}[h]
:=
\Delta_\N^{\downarrow}
\cap
\Delta_\N^{(\alpha)}[h]
$.

For $n \in \N$ and $\bmu \in \Delta_\N$,
we write
$\X = (X_1, \dots, X_n) \sim \bmu^{n}$
to mean that the components of the vector $\X$
are drawn iid from $\bmu$.
The empirical measure
$\empdist \in \Delta_\N$
induced by
the sample $\X$
is defined
by
$
\empdist(j) = \frac{1}{n}\sum_{i \in [n]}\pred{X_i = j}$.
For any $\bxi\in\R^\N$ and $0< p<\infty$,
the $\ell_p$ (pseudo)norm is defined by
$\nrm{\bxi}_p^p=\sum_{j\in\N}|\bxi(j)|^p$
and
$\nrm{\bxi}_\infty=\sup_{j\in\N}|\bxi(j)|$.

For $\alpha
,
h>0$, and $n\in\N$,
define the $L_1$ {\em minimax risk}
for the $\alpha$th moment by
\beqn
\label{eq:risk}
\risk_n^{(\alpha)}(h)
\eqdef
\inf_{\hat{H}} \sup_{\bmu 
\in\Delta_\N^{(\alpha)}[h]
} \mexp|\hat{H}(X_1, \dots, X_n) - \ent(\bmu)|,
\eeqn
where the infimum is over all
mappings
$\hat{H} \colon \N^n \to \R_+$.

\section{Main results}
\label{sec:main-res}

Our first result is a dimension-free analogue
of \eqref{eq:CT}:
\begin{theorem}
\label{thm:dimfree}  
For
all
$\alpha>1$,
$\ent:\Delta_\N^{(\alpha)}\to\R_+$
is uniformly continuous under $\ell_1$.
In particular,
for all
$\bmu, \bnu \in \Delta_\N^{(\alpha)}$
satisfying
$
\nrm{\bmu
- \bnu
}
_\infty
<
\mathe^{-\alpha}
$,
we have
\beqn
\label{eq:main-ub}
\abs{\ent(\bmu) - \ent(\bnu)}
&\leq&
\nrm{\bmu - \bnu}_1^{1 - 1/\alpha}  \left(
2\mathe^\alpha
\nrm{\bmu-\bnu}_\infty
\log^\alpha\frac{1}{\nrm{\bmu-\bnu}_\infty}
+ \ent[\alpha](\bmu)
+ \ent[\alpha](\bnu)
\right)^{1/\alpha}\\
&\leq&
\label{eq:main-ub-mod}
\nrm{\bmu - \bnu}_1^{1 - 1/\alpha}  
\left(2
\mathe
\nrm{\bmu-\bnu}_\infty^{1/\alpha}
\log\frac{1}{\nrm{\bmu-\bnu}_\infty}
+ \ent[\alpha](\bmu)^{1/\alpha}
+ \ent[\alpha](\bnu)^{1/\alpha}
\right)
.
\eeqn
Moreover,
a weaker form of
(\ref{eq:main-ub})
holds with
$\alpha^\alpha$
in place of
$
\mathe^\alpha
\nrm{\bmu-\bnu}_\infty
\log^\alpha\frac{1}{\nrm{\bmu-\bnu}_\infty}
$
under the weaker condition
$
\nrm{\bmu
- \bnu
}
_\infty
<
1/2
$,
from which
the correspondingly
weaker form of (\ref{eq:main-ub-mod})
follows as well.
\end{theorem}

The
requirement 
in Theorem~\ref{thm:dimfree} 
that $\alpha>1$
cannot be dispensed with,
as
the function $\ent:\Delta_\N^{(\alpha)}[h]\to\R_+$
is not continuous under $\ell_1$
for $\alpha=1$
(see Remark
following Lemma~\ref{lem:Hcont}),
and,
a fortiori, is not uniformly continuous.
Thus, there can be no
function $F:\R_+^2\to\R_+$
  satisfying
  \beq
\abs{\ent(\bmu) - \ent(\bnu)}
&\leq&
F(\nrm{\bmu-\bnu}_1,
h
),
\qquad
h>0,
\bmu,\bnu\in
\Delta_\N^{(1)}[h]
\eeq
with the additional property that
for any two sequences $\bmu_n,\bnu_n\in\Delta_\N$
satisfying
$\eps_n:=\nrm{\bmu_n-\bnu_n}_1\to0$,
it holds that
$F(\eps_n,h)\to0$.
Moreover, the upper bound in 
Theorem~\ref{thm:dimfree} is tight, up to a constant factor:
\begin{theorem}
\label{thm:dimfree-lb}
For every
$0<\eps<
1/2
$
and
$\alpha\ge1$,
there are
$\bmu,\bnu\in
\Delta_\N
$
such that
$\eps=\nrm{\bmu - \bnu}_1
\ge\nrm{\bmu - \bnu}_\infty
$
and
\beqn
\label{eq:lb-tight}
\abs{\ent(\bmu) - \ent(\bnu)}
&\ge&
c
\eps^{1 - 1/\alpha}
\paren{
2\mathe\eps^{1/\alpha}\log\frac1\eps
+ \ent[\alpha](\bmu)^{1/\alpha}
+ \ent[\alpha](\bnu)^{1/\alpha}
},
\eeqn
where $c\geq \frac{1}{2\mathe+2}$
is a universal constant.
\end{theorem}

Perhaps surprisingly,\footnote{
Since $\ell_1$ dominates all of the $\ell_p$
norms, continuity of a function under $\ell_p$
trivially implies continuity under $\ell_1$,
but the reverse implication is generally not true.
}
it turns out that
$\ent:\Delta_\N^{(\alpha)}[h]
\to\R_+$
is uniformly continuous 
not only under $\ell_1$,
but actually
under all $\ell_p$
norms:
\begin{theorem}
  \label{thm:Lp-pos}
  There is a function
  $F:\R_+^4\to\R_+$
  such that
  \beq
  \abs{\ent(\bmu) - \ent(\bnu)}
  &\leq&
  F(\nrm{\bmu-\bnu}_p,h,\alpha,p),
\qquad
  h>0,\alpha>1,p\in[1,\infty], \bmu,\bnu\in\Delta_\N^{(\alpha)}[h]
  \eeq
  with the additional property that
whenever $\eps_n:=\nrm{\bmu_n-\bnu_n}_p\to0$,
we have
$F(\eps_n,h,\alpha,p)\to0$.

\end{theorem}
\noindent Remark.
Although 
Theorem~\ref{thm:Lp-pos}
establishes uniform continuity,
it gives
no hint 
as to the functional dependence of the modulus of continuity $F$
on $\alpha$, $p$, $h$, and $\nrm{\bmu-\bnu}_p$.
We leave this as a fascinating open problem --- even though
the practical applications are likely to be limited:
it follows from \citet{wyner2003lower}
and Theorem~\ref{thm:plug-in-minimax}
that for $p=\alpha=2$ and fixed $h$,
$F(\nrm{\bmu-\bnu}_2,h,2,2)$
cannot decay at a faster rate than $1/\log(1/\nrm{\bmu-\bnu}_2)$.
 
Combining Theorem~\ref{thm:dimfree}
with (\ref{eq:ckw20})
yields an empirical (under moment assumptions) bound
for the plug-in entropy estimator:
\begin{corollary}
\label{cor:adaptive-hybrid-bound}
For 
all
$\alpha>1$,
$h>0$,
$\delta\in(0,1)$,
$n \geq  2 \log \frac{4}{\delta}$,
and
$\dist \in \Delta_\N^{(\alpha)}[h]$, 
we have that
\beqn
\label{eq:AHB-crude}
\abs{\ent(\bmu)-\ent(\empdist)}
&\le&
\Big( 2 \alpha^{\alpha} + 
h
+ \ent[\alpha](\empdist) \Big)^{1/\alpha} 
\left(
\frac{2\nrm{\empdist}_{1/2}^{1/2}}{\sqrt{n}} + 
6 \sqrt{\frac{\log{(
4
/\delta)}}{2n}}
\right)^{1 - 1/\alpha}
\eeqn
holds
with probability at least $1 - \delta$.
For 
all
$\alpha>1$,
$h>0$,
$0 < \eps <\mathe^{-\alpha}$,
$\delta\in(0,1)$,
$n \geq  \frac{2}{\eps^2} \log \frac{4}{\delta}$,
and
$\dist \in \Delta_\N^{(\alpha)}[h]$,
we have that
\beqn
\abs{\ent(\bmu)-\ent(\empdist)}
&\le&
\Big( 2\mathe^{\alpha}\eps\log^{1/\alpha}\frac1\eps 
\\\nonumber
&
+ 
&
\label{eq:AHB-finer}
h
+ \ent[\alpha](\empdist) \Big)^{1/\alpha} 
\left(
\frac{2\nrm{\empdist}_{1/2}^{1/2}}{\sqrt{n}} + 
6 \sqrt{\frac{\log{(
4
/\delta)}}{2n}}
\right)^{1 - 1/\alpha}
\eeqn
holds
with probability at least $1 - \delta$.
For 
all
$\alpha>1$,
$h>0$,
$\delta\in(0,1)$,
$n \geq  {2e^{2 \alpha}}{} \log \frac{4}{\delta}$,
and
$\dist \in \Delta_\N^{(\alpha)}[h]$,
we have that
\beqn
\abs{\ent(\bmu)-\ent(\empdist)}
&\le&
\Big( 2 \left(\frac{\mathe}{2}\right)^\alpha  \sqrt{\frac{2}{n}\log \frac{4}{\delta}}
\abs{\log \left(\frac{2}{n} \log \frac{4}{\delta}\right)}^{\alpha}
\\\nonumber
&+&
h
\label{eq:AHB-finer'}
+ \ent[\alpha](\empdist) \Big)^{1/\alpha} 
\left(
\frac{2\nrm{\empdist}_{1/2}^{1/2}}{\sqrt{n}} + 
6 \sqrt{\frac{\log{(
4
/\delta)}}{2n}}
\right)^{1 - 1/\alpha}
\eeqn
holds
with probability at least $1 - \delta$.
\end{corollary}
\noindent Remark.
Since the estimates in
Corollary~\ref{cor:adaptive-hybrid-bound}
involve the random quantity
$\nrm{\empdist}_{1/2}$, it is natural
to inquire as to the behavior of the latter.
It follows from 
\citet[Proposition C.1]{CKW20}
that
$n^{-1/2}\nrm{\empdist}_{1/2}^{1/2}
\to0$
in 
almost surely.
The rate of convergence must necessarily
depend on $\bmu$ itself (cf. 
\citet[Remark 9]{berend2013sharp}).

In Section~\ref{sec:rates},
we compare the rates implied by
Corollary~\ref{cor:adaptive-hybrid-bound}
to the state of the art on various distributions.

Next, we 
examine
the optimality of the plug-in estimate
by analyzing the minimax risk, defined in
(\ref{eq:risk}).
It was known 
\citep[Appendix~A]{silva2018shannon}
that assuming $\ent(\bmu)<\infty$
does not suffice to yield a minimax rate
for the $L_2$ risk:
\beq
\inf_{\hat{H}:\N^n\to\R_+} 
\sup_{
\bmu 
\in\Delta_\N^{(1)}
}
\mexp
\paren{
\hat{H}(X_1, \dots, X_n) - \ent(\bmu)
}
^2
  = \infty.
\eeq
This technique yields
an analogous result for the $L_1$ risk
as well.
We 
strengthen these results
in two ways:
(i) by lower-bounding the $L_1$ risk
(rather than $L_2$, which is never smaller),
and (ii)
by
restricting $\bmu$ to $\Delta_\N^{(1)}[h]$
and obtaining a finitary, quantitative lower bound:

\begin{theorem}
  \label{thm:no-emp}

For $\alpha=1$,
there is a universal constant $C>0$ such that
for all $h>1$ and $n\in\N, n \geq 2$, we have
$\risk_n^{(1)}(h)\ge Ch$.

\end{theorem}
\noindent Remark.
The above result complements --- but is not directly
comparable to ---
\citet[Theorem 4]{antos2001convergence}.
Ours gives a quantitative dependence on $h$ but
constructs an adversarial distribution for
each sample size $n$; theirs is asymptotic only
but a single adversarial distribution suffices
for all $n$.

\noindent Remark.
Our technique immediately yields a lower bound of
$C h^2$
on the $L_2$ minimax risk.

In contradistinction to the 
$\alpha=1$ case, where no minimax rate
exists, we show that the plug-in estimator
is minimax for all $\alpha>1$:
\begin{theorem}
\label{thm:plug-in-minimax}
The following bounds hold
for the
$L_1$ minimax risk:
\begin{enumerate}
\item[(a)]
Upper bound: for all $h>0,\alpha>1$,
\beq
\risk_n^{(\alpha)}(h)
&\leq&
\frac{1+\log n}{\sqrt{n}} + \frac{2^{\alpha-1}h}{\log^{\alpha-1}n},
\qquad n\in\N
;
\eeq
further, this
bound is achieved by the plug-in estimate
$\ent(\empdist)$.
\item[(b)]
Lower bound:
for each $\alpha > 0$, $n\in \N$ there is an 
$h > 0$ such that
\beq
\risk_{n}^{(\alpha)}(h)
&\geq&
\frac{h}{4 \cdot 3^\alpha \log^{\alpha-1}n}.
\eeq
\end{enumerate}
\end{theorem}
\paragraph{Open problem.}
Close the gap in the dependence on
$\alpha$ in the upper and lower bounds.\\
\paragraph{Open problem.}
Another gap between the upper and lower bounds
is the quantified on $h$:
in the upper bound, it is ``for all'', while
in the lower bound, it is ``exists''.
Closing this gap is also of interest.
\\

Finally, in Section~\ref{sec:conjectures},
we resolve most of the conjectures
posed by \citet{jurgensen2010entropy}.

\section{Related work}
\label{sec:rel-work}

\paragraph*{Continuity, convergence, moments of information}
\citet{Zhang07} gave a
sharpened version of
(\ref{eq:CT})
and
\citet{ho2010interplay}
presented analogous bounds;
\citet{Audenaert07}
proved a non-commutative generalization.
\citet[Theorem 5]{DBLP:journals/tit/Sason13}
upper-bounds $|\ent(\bmu)-\ent(\bnu)|$
in terms of quantities related to $\nrm{\bmu-\bnu}_1$,
where (at most) one of them
is allowed to have infinite support.
Even though $\ent(\cdot)$ is not
continuous on $\Delta_\N$,
the plug-in estimate $\ent(\empdist)$
converges to $\ent(\bmu)$
almost surely 
and in $L_2$ 
\citep{antos2001convergence}.
\citet{silva2018shannon} studied a variety of restrictions on 
distributions over infinite alphabets
to derive strong consistency results and rates of convergence. Moments of information were apparently first defined in \citet{1053843}.

\paragraph*{Entropy estimation}
Recent surveys of entropy estimation
results may be found in 
\citet{jiao2015minimax,verdu2019empirical}. 
The 
finite-alphabet case
is particularly well-understood. 
For fixed 
alphabet
size
$d<\infty$, the plug-in estimate
is asymptotically minimax optimal
\citep{paninski2003estimation}.
\citet{paninski2004estimating}
non-constructively
established the existence
of
a sublinear (in $d$) entropy estimator.
The optimal dependence on
$d$ (at fixed accuracy) 
was settled by 
\citet{valiant2011clt,valiant2017estimating}
as being $\Theta(d/\log d)$.

The $\Theta(d/\log d)$
dependence on the 
alphabet
size
is also relevant 
in the so-called
{\em high dimensional} asymptotic
regime, where $d$ grows with $n$.
Here, the plug-in estimate is no longer optimal, and more sophisticated
techniques are called for
\citep{valiant2011clt,valiant2011power,valiant2017estimating}. 
The works of \citet{wu2016minimax,jiao2015minimax,
han2015adaptive,
jiao2017maximum} characterized the minimax rates 
for the high-dimensional regime:
a small additive error of $\eps$
requires $\Theta(d/\eps\log d)$
samples.
Building off of these polynomial-approximation based constructions, \citet{acharya2017unified} design an additional optimal estimator, this one based on a profile maximum likelihood approach that can also estimate a variety of other important statistics.
\citet{fukuchi2017minimax,fukuchi2018minimax} generalize the optimal estimators to estimate any additive functional, recovering in particular the optimal rates for entropy.
\citet{acharya2019estimating} modify these optimal estimators with the added goal of low space complexity.

Finally, there is the 
infinite-alphabet case.
Although here the plug-in estimate is
again universally
strongly
consistent,
control of the convergence rate
requires some assumption on the
sampling distribution
--- and
\citet{antos2001convergence}
compellingly argue that moment assumptions
are natural and minimalistic.
Absent any prior assumptions,
the 
$L_1$
(and hence 
$L_2$) convergence rate
of {\em any} estimator
can be made arbitrarily slow
(Theorem 4 ibid.).
The present paper proves
a variant of this result
(see Theorem~\ref{thm:no-emp}
and the Remark following it).
\citet{antos2001convergence} 
further show that
even under moment assumptions, there is no polynomial rate of convergence for the plug-in estimate:
there is no $\beta>0$ 
such that its risk decays as
$O(n^{-\beta})$. 
\citet{wyner2003lower} showed that
the plug-in estimate achieves
a rate of
$O(\frac{1}{\log n})$
for bounded second moment,
and this is minimax optimal.
\citet{DBLP:conf/soda/BrautbarS07}
exhibited a function of the 
higher moments
that can be used in place of %
alphabet
size to give a multiplicative approximation to the entropy.

The empirical nature of Corollary~\ref{cor:adaptive-hybrid-bound} can be seen as a distribution-dependent improvement over otherwise worst-case minimax guarantees. It can be compared, in this light, to the ``instance-optimality" program of \citet{hao2018unified,hao2020data} and the adaptive guarantees of \citet{han2015adaptive}.

\section{Proofs}
\label{sec:proofs}

\subsection{Proof of Theorem~\ref{thm:dimfree}}

We begin with a subadditivity result for
the $\alpha$th moment of information
(which we state for $\alpha>0$,
even though only the range $\alpha>1$
will be needed).

\begin{lemma}
\label{lemma:decomposition-entropy-of-difference-measure}
For $\alpha > 0$ and $\bmu, \bnu \in \Delta^{(\alpha)}_{\N}$,
we have
\beqn
\label{eq:subadd-crude}
\ent[\alpha](\abs{\bmu - \bnu}) 
&\le&
2\alpha^\alpha + \ent[\alpha](\bmu) + \ent[\alpha](\bnu).
\eeqn
If, additionally,
$\nrm{\bmu-\bnu}_\infty\le\mathe^{-\alpha}$,
then
\beqn
\label{eq:subadd-sharp}
\ent[\alpha](\abs{\bmu - \bnu}) 
&\le&
2\mathe^\alpha
\nrm{\bmu-\bnu}_\infty
\log^\alpha\frac{1}{\nrm{\bmu-\bnu}_\infty}
+ \ent[\alpha](\bmu) + \ent[\alpha](\bnu).
\eeqn
\end{lemma}
\begin{proof}
Define $\loch[\alpha] \colon [0,1] \to \R_+$
by $z \mapsto z \log^{\alpha}(1 /z)$, 
where 
$\loch[\alpha](0)=0$.
The function
$\loch[\alpha]$ is increasing on $[0, \mathe^{- \alpha}]$ and decreasing on $[\mathe^{-\alpha}, 1]$.
The maximum is therefore achieved at $z = \mathe^{- \alpha}$, and 
\beqn
\label{eq:maxhalpha}
\max_{z \in [0, 1]} \loch[\alpha](z) = \loch[\alpha](\mathe^{- \alpha}) = \mathe^{-\alpha} \alpha^\alpha.
\eeqn
Now decompose $\ent[\alpha]$:
\beqn
\label{eq:entdiffdecomposition}
\ent[\alpha](\abs{\bmu - \bnu}) 
&=& 
\sum_{
i:\bmu(i) \lor \bnu(i) > \mathe^{-\alpha} 
} \loch[\alpha]( \abs{\bmu(i) - \bnu(i)}) 
 + \sum_{
i:\bmu(i) \lor \bnu(i) \leq \mathe^{-\alpha} 
} \loch[\alpha](\abs{ \bmu(i) - \bnu(i)}). 
\eeqn
To prove the lemma, we bound the two terms of \eqref{eq:entdiffdecomposition} separately. The second term can be bound in two ways, yielding \eqref{eq:subadd-crude} and \eqref{eq:subadd-sharp}, respectively. To bound
the first term of \eqref{eq:entdiffdecomposition}, 
notice that
$\bmu \in \Delta_\N$
implies
that $\abs{ \set{i \in \N \colon \bmu(i) > \mathe^{-\alpha}} } \leq \mathe^{\alpha}$, 
and similarly for $\bnu$. 
Thus,
\beqn
\label{eq:1st-term-subadd}
\sum_{
i:\bmu(i) \lor \bnu(i) > \mathe^{-\alpha} 
} \loch[\alpha]( \abs{\bmu(i) - \bnu(i)})
&\le&
\big(\abs{ \set{i 
\colon \bmu(i) > \mathe^{-\alpha}} } 
+ \abs{ \set{i 
\colon \bnu(i) > \mathe^{-\alpha}} } \big) \max_{z \in [0,1]} \loch[\alpha](z) 
\\&\le&  
\nonumber
2\mathe^{\alpha} \mathe^{-\alpha} \alpha^\alpha = 2 \alpha^\alpha
.
\eeqn
For the second term of \eqref{eq:entdiffdecomposition}, notice that when $\bmu(i) \lor \bnu(i) \leq \mathe^{-\alpha}$, 
the monotonicity of $\loch[\alpha]$
implies
\beq
\loch[\alpha]( \abs{\bmu(i) - \bnu(i)}) \leq \loch[\alpha]( \bmu(i) \lor \bnu(i)),
\eeq
and hence
\beq
\sum_{
i \in \N :
\bmu(i) \lor \bnu(i) \leq \mathe^{-\alpha} 
} 
\loch[\alpha](\abs{ \bmu(i) - \bnu(i)})
&\leq& 
\sum_{
i: 
\in 
\bmu(i) \lor \bnu(i) \leq \mathe^{-\alpha} 
} 
\loch[\alpha]( \bmu(i) \lor \bnu(i)) \\
&\leq& 
\sum_{
i: 
\in 
\bmu(i) \lor \bnu(i) \leq \mathe^{-\alpha} 
} 
\loch[\alpha]( \bmu(i)) + \loch[\alpha]( \bnu(i)) \\
&\leq& \ent[\alpha](\bmu) + \ent[\alpha](\bnu);
\eeq
this proves \eqref{eq:subadd-crude}.
Given the additional condition 
$\nrm{\bmu-\bnu}_\infty\le\mathe^{-\alpha}$,
to prove \eqref{eq:subadd-sharp},
put $\eps=\nrm{\bmu-\bnu}_\infty$
and
modify \eqref{eq:1st-term-subadd} as
follows:
\beq
\sum_{
i:\bmu(i) \lor \bnu(i) > \mathe^{-\alpha} 
} \loch[\alpha]( \abs{\bmu(i) - \bnu(i)})
&\le&
\big(\abs{ \set{i 
\colon \bmu(i) > \mathe^{-\alpha}} } 
+ \abs{ \set{i 
\colon \bnu(i) > \mathe^{-\alpha}} } \big) \loch[\alpha](\eps)
\\
&\le&
2\mathe^\alpha\loch[\alpha](\eps).
\eeq
The latter holds since
$\abs{\bmu(i) - \bnu(i)}
\le
\nrm{\bmu-\bnu}_\infty
$,
and so
$\loch[\alpha]( \abs{\bmu(i) - \bnu(i)})
\le
\loch[\alpha](\eps)
$,
by $\loch[\alpha]$'s
monotonicity on $[0,\mathe^{-\alpha}]$.
\end{proof}

\begin{proof}[Proof of Theorem~\ref{thm:dimfree}]
The concavity argument in the proof of 
\citet[Theorem~17.3.3]{cover2001elements}, 
immediately implies
\beq
\abs{ \ent(\bmu) -  \ent(\bnu)} 
&\leq& \ent(\abs{\bmu - \bnu}).
\eeq
Then, via an application of H\"older's inequality,
\beq
\ent(\abs{\bmu - \bnu})
&=& \sum_{i\in\N}\abs{\bmu(i) - \bnu(i)}\log \frac{1}{\abs{\bmu(i) - \bnu(i)}} \\
&=& \sum_{i\in\N} \abs{\bmu(i) - \bnu(i)}^{1-1/\alpha}\cdot\abs{\bmu(i) - \bnu(i)}^{1/\alpha} 
\log \frac{1}{\abs{\bmu(i) - \bnu(i)}}\\
&\le&
\paren{\sum_{i \in \N} \left(\abs{\bmu(i) - \bnu(i)}^{1-1/\alpha}\right)^{1/(1-1/\alpha)}
}^{1-1/\alpha} 
\paren{\sum_{i \in \N} \left(\abs{\bmu(i) - \bnu(i)}^{1/\alpha} \log{ \frac{1}{\abs{\bmu(i) - \bnu(i)}}}\right)^\alpha}^{1/\alpha}\\
&=& \nrm{\bmu - \bnu}_1^{1-1/\alpha}
\ent[\alpha](\abs{\bmu - \bnu})^{1/\alpha}.
\eeq
The claim follows by
invoking Lemma~\ref{lemma:decomposition-entropy-of-difference-measure}
and the subadditivity of $t\mapsto t^{1/\alpha}$
for $t\ge0$ and $\alpha>1$.
\end{proof}

\subsection{Proof of Theorem~\ref{thm:dimfree-lb}}

Let
$0<\eps
\le
1/2
$
be given and
choose
$\bmu,\bnu\in
\Delta_\N
$
as follows:
$\bmu=(1,0,0,\ldots)$
and $\bnu=(1-\eps,\eps,0,0,\ldots)$.
Then 
left-hand side
of \eqref{eq:lb-tight}
is $L(\eps)=\ent(\nu)$:
\beq
L(\eps) &=&
(1-\eps)\log\frac{1}{1-\eps}
+
\eps\log\frac{1}{\eps} =:
L_1(\eps)+L_2(\eps)
,
\eeq
and
right-hand side
of \eqref{eq:lb-tight},
without the constant $c$,
is
\beq
R(\eps) &=&
\eps^{1-1/\alpha}
\paren{
2\mathe\eps^{1/\alpha}\log\frac1\eps
+
\paren{
(1-\eps)\abs{\log\frac{1}{1-\eps}}^\alpha
+
\eps\abs{\log\frac{1}{\eps}}^{\alpha}
}^{1/\alpha}
}\\
&\leq&
\eps^{1-1/\alpha}
\paren{
2\mathe\eps^{1/\alpha}\log\frac1\eps
+
2^{1/\alpha} \max \set{
(1-\eps)^{1/\alpha}\log\frac{1}{1-\eps}
,
\eps^{1/\alpha}\log\frac{1}{\eps}
}
}\\
&\leq&
2\eps^{1-1/\alpha}(1-\eps)^{1/\alpha}\log\frac{1}{1-\eps}
+
(2\mathe+2)\eps\log\frac1\eps
\\
&=:&
R_1(\eps)+R_2(\eps)
.
\eeq

Now
$L_1(\eps) \ge R_1(\eps)/2$ 
for $\eps\in(0,\frac{1}{2}]$ and
$L_2(\eps)=
\frac{R_2(\eps)}{2\mathe+2}$,
and
therefore 
$
L_1(\eps)+L_2(\eps)\ge R_1(\eps)/2+\frac{1}{2\mathe+2}R_2(\eps)\ge \frac{1}{2\mathe+2}R(\eps)$.\QED

\subsection{Proof of Theorem~\ref{thm:Lp-pos}}

The following 
fact 
\citep[Theorem 3.5 and Eq. (5) on p. 83]{MR1817225}
will be useful\footnote{
The result is stated for functions
in $f\in L_2(\R^n)$ 
and their
symmetric-decreasing rearrangements
$f^*$,
but the specialization
to discrete distributions is straightforward.
We convert $\bmu$
to a function $f:\R_+\to\R_+$
via $f(x)=\mu(\ceil{x})$
and $\bnu$ to $g(x)$ analogously.
A direct calculation then 
shows
that 
$
\norm{\bmu-\bnu}_p
=
\norm{f-g}_p
$
and
$
\norm{\decr{\bmu}-\decr{\bnu}}_p
=
\norm{f^*-g^*}_p
$, to which the result from
\citet{MR1817225} applies
to yield
(\ref{eq:munudecr}).
}:
\beqn
\label{eq:munudecr}
\norm{\decr{\bmu}-\decr{\bnu}}_p
&\le&
\nrm{{\bmu}-{\bnu}}_p,
\qquad p\in[1,\infty],
\;
\bmu,\bnu\in\Delta_\N.
\eeqn

A result of \citet{MR21585} (more accurately credited to
Riesz, 1928 \citep{MR2728440})
implies that a sequence $\set{\bxi_{n\in\N}}\subset\ell_1(\N)$
converging pointwise to some $\bxi\in\ell_1(\N)$
also converges in $\ell_1$ iff $\nrm{\bxi_n}_1\to\nrm{\bxi}_1$.
This immediately implies
\begin{lemma}
\label{lem:pt-l1}
If $\set{\bmu_{n\in\N}}\subset\Delta_\N$ converges pointwise to some
$\bmu\in\Delta_\N$, then it also converges in $\ell_1$.
\end{lemma}

\citet[Lemma 1]{BKZ17} showed that
$\Delta_\N^{\downarrow(1)}[h]$
is compact under $\ell_1$.
We begin by extending
this result
to
general $\alpha,p$.
\begin{lemma}
  \label{lem:compact}
  For all $\alpha\ge1$, $p\in[1,\infty]$, and $h>0$,
  the set $\Delta_\N^{\downarrow(\alpha)}[h]$ is compact under $\ell_p$.
\end{lemma}
\noindent Remark.
This is quite false if 
either
the non-increasing 
or the bounded-entropy
condition
is omitted. 
For a counterexample to the former,
consider
the sequence $\bmu_n\in\Delta_\N$
defined by $\bmu_n(i)=\pred{i=n}$.
For a counterexample to the latter,
consider
the sequence $\bmu_n\in\Delta_\N$,
where $\bmu_n$ is uniform on $[n]$.
\begin{proof}
We closely follow the proof strategy of
\citet[Lemma 1]{BKZ17}.
In a metric space, compactness and sequential compactness are equivalent.
Let $\bmu_{n\in\N}$ be a sequence in
$\Delta_\N^{\downarrow(\alpha)}[h]$.
Since $[0,1]$ is compact, every $\set{\bmu_n(i):n\in\N}$
has a convergent subsequence, and hence
$\bmu_{n\in\N}$ has a pointwise convergent subsequence.
There is thus no loss of generality in assuming that
$\bmu_n\to\bmu$ pointwise. Obviously, $\bmu$ is non-negative and
non-increasing.
It remains to show that
\begin{itemize}
\item[(a)] $\sum_{i\in\N}\bmu(i)=1$,
\item[(b)] $\ent[\alpha](\bmu)\le h$,
\item[(c)] $\nrm{\bmu_n-\bmu}_p\to0$.
\end{itemize}

To show (a), assume, for a contradiction, that
$\sum_{i\in\N}\bmu(i)>1$. Then there must be
an $i_0\in\N$ such that $\sum_{i=1}^{i_0}\bmu(i)>1$.
But the latter must then hold for all $\bmu_n$ with $n$
sufficiently large, which contradicts $\bmu_n\in\Delta_\N$.
Now assume
$\eps:=1-\sum_{i\in\N}\bmu(i)>0$.
For any $i_0\in\N$, we have
$\sum_{i=1}^{i_0}\bmu_n(i)<1-\eps/2$
for all sufficiently large $n$.
Now every $\bnu\in
\Delta_\N^{\downarrow}
$ satisfies
$
\bnu(i)\le \frac1i(\bnu(1)+\bnu(2)+\ldots+\bnu(i))\le
\frac1i
.
$
Hence,
\beq
\sum_{i=i_0+1}^\infty\bmu_n(i)\abs{\log\bmu_n(i)}^\alpha
\ge
\sum_{i=i_0+1}^\infty\bmu_n(i)(\log i_0)^\alpha
>
\frac\eps2(\log i_0)^\alpha.
\eeq
Choosing $i_0$ sufficiently large makes the latter expression exceed
$h$, violating the assumption $\bmu_n\in
\Delta_\N^{\downarrow(\alpha)}[h]
$. Thus (a) holds.

To show (b), assume for a contradiction that
$\ent[\alpha](\bmu)>h$ --- and, in particular,
$\sum_{i=1}^{i_0}
\bmu(i)\abs{\log\bmu(i)}^\alpha
>h$ for some $i_0\in\N$.
But the latter must
hold for all $\bmu_n$ with $n$ sufficiently large,
a contradiction.

Finally, to show (c), 
we invoke Lemma~\ref{lem:pt-l1}:
if $\set{\bmu_{n\in\N}}\subset\Delta_\N$ converges pointwise to some
$\bmu\in\Delta_\N$, then it also converges in $\ell_1$.
Since $\ell_1$ dominates every $\ell_p$, $p>1$, this proves (c).
\end{proof}

Next, we examine the continuity of $\ent(\cdot)$ on
$\Delta_\N^{\downarrow(\alpha)}[h]
$ under $\ell_p$.

\begin{lemma}
\label{lem:Hcont}
Fix $h>0$, $\alpha>1$, and $p\in[1,\infty]$.
If $\set{\bmu_{n\in\N}}\subset
\Delta_\N^{\downarrow(\alpha)}[h]
$ converges in $\ell_p$,
then its limit is some $\bmu\in\Delta_\N^{\downarrow(\alpha)}[h]$
and 
furthermore,
$\ent(\bmu_n)\to\ent(\bmu)$.
In other words, $\ent(\cdot)$
is continuous on
$\Delta_\N^{\downarrow(\alpha)}[h]
$
under $\ell_p$.
\end{lemma}  
\noindent Remark.
We note that $\ent(\cdot)$ is not
continuous on 
$\Delta_\N^{\downarrow(1)}[h]$
under $\ell_p$, $p\in[1,\infty]
$,
as evidenced by the sequence $\bmu_n=(1-\eps_n,\eps_n/n,\ldots,\eps/n,0,0,\ldots)$,
with support size $n+1$.
We can choose $\eps_n$ so that $\ent(\bmu_n)=h$,
but of course the limiting $\bmu$ has $\ent(\bmu)=0$
(see Example 1 in \citet{BKZ17}).
\begin{proof}
It follows from Lemma~\ref{lem:compact}
that 
the limiting
  $\bmu$ belongs to $\Delta_\N^{\downarrow(\alpha)}[h]$.
  Further,
  Lemma~\ref{lem:pt-l1}
implies that $\bmu_n\to\bmu$ in $\ell_1$.
  Invoking the continuity result
  in Theorem~\ref{thm:dimfree}
  proves the claim. 
\end{proof}

\begin{proof}[Proof of Theorem~\ref{thm:Lp-pos}]
  It follows from Lemma~\ref{lem:Hcont}
  that
$\ent(\cdot)$
is continuous on
$\Delta_\N^{\downarrow(\alpha)}[h]
$
under $\ell_p$.
Since, by Lemma~\ref{lem:compact},
$\Delta_\N^{\downarrow(\alpha)}[h]
$
is compact under
$\ell_p$, it follows that $\ent(\cdot)$
is uniformly continuous on 
$\Delta_\N^{\downarrow(\alpha)}[h]
$: there is a function $F$
such that
\beq
\abs{\ent(\bmu) - \ent(\bnu)}
&\leq&
F(\nrm{\bmu-\bnu}_p,h,\alpha,p),
\qquad
\bmu,\bnu\in
\Delta_\N^{\downarrow(\alpha)}[h]
\eeq
and
$\eps_n:=\nrm{\bmu_n-\bnu_n}_p\to0$
$\implies$
$F(\eps_n,h,\alpha,p)\to0$.
Now, for all 
$\bmu,\bnu\in\Delta_\N^{(\alpha)}[h]$
we have
\beq
\abs{\ent(\bmu) - \ent(\bnu)}
=
\abs{\ent(\decr{\bmu}) - \ent(\decr{\bnu})}
\le
F(\norm{\decr{\bmu}-\decr{\bnu}}_p,h,\alpha,p).
\eeq
It follows from (\ref{eq:munudecr})
that
$\nrm{\bmu_n-\bnu_n}_p\to0
\implies
\norm{\decr{\bmu}_n-\decr{\bnu}_n}_p\to0$,
which concludes the proof.
\end{proof}

\subsection{Proof of %
Corollary~\ref{cor:adaptive-hybrid-bound}
}

\begin{proof}
Fix $0 < \eps <\mathe^{-\alpha}$. 
Consider two potential ``bad'' events:
$B_1$, where
$\norm{\empdist-\bmu}_\infty> \eps$,
and $B_2$, where
$
\nrm{\bmu-\empdist}_1
>
\frac{2\nrm{\empdist}_{1/2}^{1/2}}{\sqrt{n}} + 
6 \sqrt{\frac{\log(4/\delta)}{2n}}.
$
Our assumption on the sample size $n$,
together with the
Dvoretzky-Kiefer-Wolfowitz inequality
\citep{MR1062069}, implies that
$\PR{
B_1
}\le\delta/2$
and
(\ref{eq:ckw20})
implies that
$\PR{
B_2
}\le\delta/2$.
Thus, with probability at least
$1-\delta$, neither of $B_1$ or $B_2$
occurs, and on the event where $\nrm{\bmu-\empdist}_\infty < \eps <\mathe^{-\alpha}$, we may invoke
Theorem~\ref{thm:dimfree},
from which the claims immediately follow.
\end{proof}

\subsection{Proof of Theorem~\ref{thm:no-emp}}

\newcommand{\oh}{h}

For $\oh > 1$ and $n \in \N, n \geq 2$,
put
$a_n = (1 - 1/(2n))\log(1 - 1/(2n))$
and
define the support size $S=S(h,n)$ by
$
S = \floor{(1/2n) \exp ( 2n (\oh + a_n) )}.
$
Consider the distributions $\bmu_0 = (1, 0, 0, \dots)$ and $\bmu_{n}$
defined by
$\bmu_n (1) = 1 - 1/(2n)$, and 
\beq
\bmu_n(i) = \frac{1}{2 n S},
\qquad
2 \leq i \leq 1 + S( \oh, n).
\eeq
We compute the Kullback-Leibler divergence and entropy:
\beqn
\label{eq:statistically-close}
\kl{\bmu_0}{\bmu_n} &=& 
\log \frac{1}{1 - 1/(2n)} 
 \le
\frac{1}{1 - 1/(2n)}-1
\leq \frac{1}{n} \\\nonumber
\ent(\bmu_0) &=& 0 \leq \oh
.
\eeqn
For $x \geq 2$, always $\floor{x} \geq x /2$. Additionally, from $2n a_n \geq -1$, and $\frac{1}{2n} \exp (2nh - 1) > 2$, we obtain that
$S > (1/4n) \exp ( 2n (\oh + a_n) )$, hence we also have that $h \geq \ent(\bmu_n) > \oh - \frac{1}{2n} \log 2$. 
Since 
$
\frac{1}{2x}\log 2
\le 1/2
$
on $[1,\infty)$
and 
$\oh > 1$,
it follows that
$\ent(\bmu_n) 
\geq \frac{\oh}{2}$, 
whence
$\abs{\ent(\bmu_0) - \ent(\bmu_n)} 
\geq \oh/2$.
To bound the $L_1$ minimax risk (defined in
(\ref{eq:risk})), we invoke
Markov's inequality:
\beq
\mexp|\hat{H}(X_1, \dots, X_n) - \ent(\bmu)|
&\ge& \frac{\oh}{4} \PR{|\hat{H}(X_1, \dots, X_n) - \ent(\bmu)| >  \frac{\oh}{4}}.
\eeq
It follows via Le Cam's two point method
\citep[Section~2.4.2]{tsybakov2008introduction}
that
\beq
\risk_n^{(1)}(h)
\geq \frac{\oh}{8} \mathe^{-n \kl{\bmu_0}{\bmu}}
\geq \frac{\oh}{8 \mathe},
\eeq
where the second inequality
stems from \eqref{eq:statistically-close}.

\QED

\subsection{Proof of Theorem~\ref{thm:plug-in-minimax}
}

We begin with an auxiliary lemma,
of possible independent interest.
\begin{lemma}
\label{lem:sandwich}
For all $\bmu\in\Delta_\N
$
and $n\in\N$,
we have
\beq
\ent(\bmu)
\ge
\mexp\ent(\empdist) \ge
\ent(\bmu)
-
\inf_{0<\eps<1}
\sqprn{
\sum_{i\in\N:\bmu(i)<\eps}
\bmu(i)\log\frac1{\bmu(i)}
+
\log\paren{1+\frac1{\eps n}}
}.
\eeq
\end{lemma}
\begin{proof}
The first inequality follows from
Jensen's, since $\ent(\cdot)$
is concave and $\mexp\empdist=\bmu$.
To prove the second inequality, 
choose $\eps>0$,
put
$J:=\set{i\in\N:\bmu(i)<\eps}$,
and compute
\beq
\mexp{\ent(\empdist)} 
& = & 
\E{
\sum_{
i
\in\N\setminus J
} \empdist(i) \log \frac{1}{\empdist(i)} 
+ 
\sum_{
i\in J
} \empdist(i) \log \frac{1}{\empdist(i)} 
}\\
&\ge&
\E{
\sum_{
i
\in\N\setminus J
} \empdist(i) \log \frac{1}{\empdist(i)} 
+
\paren{\sum_{
i\in J
} \empdist(i)}\log \frac{1}{
      \sum_{
i\in J
      } \empdist(i)
      } } \\
&=:&
\mexp\ent(\tilde\bmu_n),
\eeq
where 
$\tilde\bmu_n$
is the ``collapsed''
version of $\empdist$,
where all of the masses
in $J
$
have been replaced
by 
a single mass equal to their sum,
and
the inequality 
holds because conditioning reduces entropy \citep[Eq.(2.157)]{cover2001elements}.
We observe that 
$\tilde\bmu_n$
has support size at most $1+1/\eps$
and invoke
\citet[Proposition~1]{paninski2003estimation}:
\beqn
\label{eq:paninski-bias}
\mexp\ent(\tilde\bmu_n)
&\ge&
\ent(\tilde\bmu)
-\log\paren{1+\frac1{\eps n}},
\eeqn
where $\tilde\bmu$
is the ``collapsed''
version of $\bmu$.
Now
\beq
\ent(\tilde\bmu)
&=&
\ent(\bmu)
+
\paren{\sum_{
i\in j
} \bmu(i)}\log \frac{1}{
\sum_{
i\in J
} \bmu(i)
}
-
\sum_{
i\in J
} \bmu(i) \log \frac{1}{\bmu(i)} 
\\&\ge&
\ent(\bmu)
-
\sum_{
i\in J
} \bmu(i) \log \frac{1}{\bmu(i)} ,
\eeq
which concludes the proof.
\end{proof}

The first part of the theorem will follow from the following proposition.

\begin{proposition}
\label{prop:plug-in-expec-bound}
For $\alpha\ge1$, $h>0$, $n\in \N$
and $\bmu\in\Delta_\N^{(\alpha)}[h]$,
we have
  \beq
\mexp{\abs{\ent(\bmu) - \ent(\empdist)}}
      &\leq&
      \frac{\log n}{\sqrt{n}} + 
      \inf_{0<\eps <1} 
      \sqprn{
       \paren{\log \frac{1}{\eps}}^{1-\alpha} %
      h
      + \log \paren{1 + \frac{1}{\eps n}} 
      }.
  \eeq
\end{proposition}
\begin{proof}
Since 
by
Lemma~\ref{lem:sandwich},
$\abs{\ent(\bmu) - \mexp{\ent(\empdist)}}
=\ent(\bmu) - \mexp{\ent(\empdist)}$,
it follows from
the triangle 
and Jensen
inequalities
that
\beqn
\mexp{\abs{\ent(\bmu) - \ent(\empdist)}} 
& \leq & 
\mexp{\abs{\ent(\empdist)- \mexp{\ent(\empdist)}}} + 
\ent(\bmu) - \mexp{\ent(\empdist)} \nonumber\\
& \leq &
\sqrt{\Var{\ent(\empdist)}} + \ent(\bmu) - \mexp{\ent(\empdist)} \nonumber\\
\label{eq:logsqrtHmu}
& \leq & \label{eq:expecbias} 
\frac{\log n}{\sqrt{n}} + \ent(\bmu) - \mexp{\ent(\empdist)},
\eeqn
where the variance bound is from \citet[Proposition~1(iv)]{antos2001estimating}.

For any $\eps>0$,
Lemma~\ref{lem:sandwich} implies
\beqn
\mexp{\ent(\empdist)}
&\geq& 
\ent(\bmu)
-
\sum_{i\in\N:\bmu(i)<\eps}
\bmu(i)\log\frac1{\bmu(i)}
-
\log\paren{1+\frac1{\eps n}}
\nonumber\\
&\geq& \ent(\bmu) - 
\paren{\log \frac{1}{\eps}}^{1-\alpha} \sum_{i\in\N:\bmu(i)<\eps} \bmu(i) \paren{\log \frac{1}{\bmu(i)}}^\alpha
- \log \paren{1 + \frac{1}{\eps n}} 
\nonumber\\
\label{eq:Hmueps}
&\ge& \ent(\bmu) - \paren{\log \frac{1}{\eps}}^{1-\alpha} \ent[\alpha](\bmu) - \log \paren{1 + \frac{1}{\eps n}} ,
\eeqn
where the second 
and third
inequalities follow
from the obvious relations
\beq
\sum_{i:\bmu(i)<\eps}
\bmu(i)\log\frac1{\bmu(i)}
&\le&
\paren{\log \frac{1}{\eps}}^{1-\alpha} \sum_{i:\bmu(i)<\eps} \bmu(i) \paren{\log \frac{1}{\bmu(i)}}^\alpha
\\&\le&
\paren{\log \frac{1}{\eps}}^{1-\alpha}
\ent[\alpha](\bmu)
.
\eeq
The claim follows by combining
(\ref{eq:logsqrtHmu})
with
(\ref{eq:Hmueps}).
\end{proof}

\begin{proof}[Proof of Theorem~\ref{thm:plug-in-minimax}(a)]
Use the fact
 that
  $\risk_n^{(\alpha)}(h) \leq \mexp{\abs{\ent(\bmu) - \ent(\empdist)}}$, invoke
  Proposition~\ref{prop:plug-in-expec-bound}
  with
$\eps = \frac{1}{\sqrt n}$ 
and 
use
$\log (1+x) \leq x$.
\end{proof}
We now prove the second half of the theorem.
\begin{proof}[Proof of Theorem~\ref{thm:plug-in-minimax}(b)]
\newcommand{\U}{\mathcal{U}}
Let $\alpha > 0$, $n\in \N$
and
define two families of distributions:
\beq
&&\U_1 
\eqdef 
\set{\bmu_1 = \Unif([n^3])}
,\\&&
\U_2 
\eqdef
\set{\bmu_2 = \Unif(A) :  A \subset [n^3] , |A|=n^2 }
.
\eeq
Let $h \eqdef 3^\alpha \log^\alpha n$ and note that $\U_1 \cup \U_2 \subseteq \Delta_\N^{(\alpha)}[h]$. Let $E$ be the event that $\X = (X_1, \dots, X_n)$ has no repeating elements, i.e $\abs{\set{X_1, X_2, \dots ,X_n}}=n$.
Let $\bmu_1 \in \U_1, \bmu_2 \in \U_2$ and consider the values $\PR[\X \sim \bmu_1^n]{E}$ and $\PR[\X \sim \bmu_2^n]{E}$. 
For $m\in\N$,
define
$K(m)$
to be the smallest $k$
such that
when uniformly throwing $m$ balls into $k$ buckets, the probability of collision is at least $1/2$.
Since $K(m)$ is known\footnote{
Better bounds exist
\citep{brink2012birthday}.
} to be at least $\sqrt{m}$ (and hence $K(n^2) > n$) we have a lower bound of $\frac{1}{2}$ on both $\PR[\X \sim \bmu_1^n]{E}$ and $\PR[\X \sim \bmu_2^n]{E}$. 
Define
$\bmu_1^n|E$
as the distribution
on $\N^n$ induced by
conditioning the product
$\bmu_1^n$ on the event $E$,
and define 
$\bmu_2^n|E$
analogously.
Our key observation is that
conditional on $E$,
$\bmu_1^n$ is uniform on
$([n^3])_n$
whereas
$\bmu_2^n=\Unif(A)^n$ is uniform on
$(A)_{n}$,
where
$
(J)_k
:=
\set{
(x_1,\ldots,x_k)\in J^k
:
|\set{x_1,\ldots,x_k}|=k
}
$
is the set of all possible ordered samples of size $k$ from a distribution supported on $J$.
To verify
this observation,
take $x=(x_1,\ldots,x_n)\in [n^3]$ and note that $(\bmu_1^n|E)(x)=\PR[\X \sim \bmu_1^n]{\set{X=x}\cap E}/\PR[\X \sim \bmu_1^n]{E}=\PR[\X \sim \bmu_1^n]{\set{X=x}}/\PR[\X \sim \bmu_1^n]{E}$ if $x$ has no repeating elements (meaning $\set{X=x} \subseteq E$) and $(\bmu_1^n|E)(x)=0$ otherwise.
Then
\beq\risk_{n}^{(\alpha)}(h)
&\geq& \inf_{\hat{H}} \sup_{\bmu \in \U_1 \cup \U_2} \E[\X \sim \bmu^n]{|\hat{H}(\X) - \ent(\bmu)|} \\
&\overset{\text{(a)}}{\geq}&
\inf_{\hat{H}} \sup_{\bmu 
\in \U_1 \cup \U_2} \E[\X \sim \bmu^n|E]{|\hat{H}(\X) - \ent(\bmu)|}\PR[\X \sim \bmu^n]{E}\\
&\geq&
\inf_{\hat{H}} \frac{1}{2}  \sup_{\bmu 
\in \U_1 \cup \U_2} \E[\X \sim \bmu^n|E]{|\hat{H}(\X) - \ent(\bmu)|} \\
&\overset{\text{(b)}}{\geq}&
\inf_{\hat{H}} \frac{1}{4}  \Bigg(\E[\X \sim \bmu_1^n|E]{|\hat{H}(\X) - \ent(\bmu_1)|} 
+ \sup_{\bmu_2 
\in \U_2} \E[\X \sim \bmu_2^n|E]{|\hat{H}(\X) - \ent(\bmu_2)|}\Bigg) \\
&\overset{\text{(c)}}{\geq}&
\inf_{\hat{H}} \frac{1}{4}  \Bigg(\E[\X \sim \bmu_1^n|E]{|\hat{H}(\X) - \ent(\bmu_1)|} 
+ \E[\bmu_2 
\sim \Unif(\U_2)] {\E[\X \sim \bmu_2^n|E]{|\hat{H}(\X) - \ent(\bmu_2)|}} \Bigg) 
\\
&\overset{\text{(d)}}{=}&
\inf_{\hat{H}} \frac{1}{4}  \Bigg(\E[\X \sim \bmu_1^n|E]{|\hat{H}(\X) - \ent(\bmu_1)|} 
+ \E[\X \sim \bmu_1^n|E]{|\hat{H}(\X) - \ent(\bmu_2)|}\Bigg) 
\\
&=& 
\inf_{\hat{H}} \frac{1}{4}  \Bigg(\E[\X \sim \bmu_1^n|E]{|\hat{H}(\X) - \ent(\bmu_1)| + |\hat{H}(\X) - \ent(\bmu_2)|}\Bigg) \\
&\overset{\text{(e)}}{\geq}& \frac{1}{4}  {|\ent(\bmu_1) - \ent(\bmu_2)|} 
=
\frac{1}{4}  \log n 
= \frac{1}{4}  \frac{h}{3^\alpha \log^{\alpha-1} n},
\eeq
where 
(a) is from the law of total expectation
(the complement of $E$ is discarded), 
(b) 
and 
(c) are bounding a supremum by an average, 
(e) is from the triangle inequality,
and 
(d) is by observing that,
by symmetry,
the operators $\E[\bmu_2 \sim \Unif(\U_2)] {\E[\X \sim \bmu_2^n|E]{\cdot}} = \E[A \sim \Unif( \sqprn{n^2} )] {\E[\X \sim \Unif((A)_n)]{\cdot}}$ and $\E[\X \sim \bmu_1^n|E]{\cdot}=\E[\X \sim \Unif( \sqprn{n^3} )]{\cdot}$ are equivalent.
(There is a minor abuse of notation in transitions after (c),
since we 
write
$\bmu_2$
without specifying
a {\em particular} member of
$\U_2$. However,
$\bmu_2$
only occurs therein
as $\ent(\bmu_2)$,
and this value is identical
for all $\bmu_2\in\U_2$.)
\end{proof}

\section{Comparative Rates}
\label{sec:rates}

Our bounds have the crucial characteristic of
being empirical (under moment assumptions). When we \textit{observe}
favorable distributions (even without a priori
knowledge of the fact), we will benefit from
tighter bounds. 
This entails some cost,
and in the worst case
our bounds will be sub-optimal. In this section,
we 
illustrate
these trade-offs for various natural classes of distributions.

For the class of all finite alphabet distributions, our bound is sub-optimal. The MLE (plug-in estimator) is competitive with the optimal estimator up to logarithmic factors in $d$, but our bounds on the MLE are loose nearly quadratically in $d/n$, in the worst case. The convergence of the empirical distribution on a finite alphabet 
in $\ell_1$ occurs at rate $\Theta(\sqrt{d/n})$, whereas the MLE entropy estimator converges at rate 
$O\left(\sqrt{\left(\frac{d}{n}\right)^2+\frac{\log^2 d}{n}}\right)$, as follows from \citet[Proposition 1]{wu2016minimax}. So any approach that
upper bounds the entropy risk
via $\ell_1$ 
(as our Theorem \ref{thm:dimfree} or Section 4 of \citet{ho2010interplay}) will be worst-case suboptimal for this class of distributions.

Nevertheless, for certain classes of distributions our bounds 
(Theorem \ref{thm:dimfree} and Corollary \ref{cor:adaptive-hybrid-bound})
can significantly outperform 
the state of the art, for small and moderate-sized
samples.
To calculate the expected rate of our approach, we apply H\"older's inequality, as in the proof of Theorem \ref{thm:dimfree}: 
\beq
\mexp{\abs{\ent(\empdist) - \ent(\dist)}} 
&\leq&
\left(\E{ 2 \alpha^\alpha + \ent[\alpha](\dist) + \ent[\alpha](\empdist) }\right)^{1/\alpha}
\left( \mexp{\ello{\empdist- \dist}}\right)^{1 - 1/\alpha} .
\eeq

Now, as in the proof of Lemma \ref{lemma:decomposition-entropy-of-difference-measure}
(recall that
$\loch[\alpha]
(z):=
z \log^{\alpha}(1 /z)$)
,
\beq
\mexp{\ent[\alpha](\empdist) } 
&=& 
\mexp{ \sum_{\substack{i \in [d] \\ \dist(i) \vee \empdist(i) \geq \mathe^{-\alpha}}}  \loch[\alpha](\empdist(i)) } + \sum_{\substack{i \in [d] \\ \dist(i) < \mathe^{-\alpha}}} \loch[\alpha](\empdist(i)) \pred{\empdist(i) <\mathe^{ - \alpha}}
\\
&\leq& 
2 \mathe^{\alpha} \max_{z \in [\mathe^{- \alpha},1]} \loch[\alpha](z)  + \sum_{\substack{i \in [d] \\ \dist(i) < \mathe^{-\alpha}}} \mexp{ \loch[\alpha](\empdist(i)) \pred{\empdist(i) <\mathe^{ - \alpha}}} \\
&=& 
2 \mathe^{\alpha} \max_{z \in [\mathe^{ - \alpha},1]} \loch[\alpha](z)  + \sum_{\substack{i \in [d] \\ \dist(i) < \mathe^{-\alpha}}} \mexp{ \loch[\alpha](\empdist(i) \pred{\empdist(i) <\mathe^{ - \alpha}}) } \\
&\stackrel{(i)}{\leq}& 
2 \mathe^{\alpha} \max_{z \in [\mathe^{ - \alpha},1]} \loch[\alpha](z)  + \sum_{\substack{i \in [d] \\ \dist(i) < \mathe^{-\alpha}} }  \loch[\alpha](\mexp{ \empdist(i) \pred{\empdist(i) <\mathe^{ - \alpha}} })  \\
&\stackrel{(ii)}{\leq}& 
2 \mathe^{\alpha} \max_{z \in [\mathe^{ - \alpha},1]} \loch[\alpha](z)  + \sum_{\substack{i \in [d] \\ \dist(i) < \mathe^{-\alpha}} }  \loch[\alpha](\mexp{ \empdist(i) })  \\
&\leq&
2 \mathe^{\alpha} \max_{z \in [\mathe^{- \alpha},1]} \loch[\alpha](z)  + \ent[\alpha](\dist) 
\stackrel{(iii)}{\leq}
2 {\alpha^\alpha}  + \ent[\alpha](\dist),
\eeq
where $(i)$ follows from Jensen's inequality, $(ii)$ is because  $\loch[\alpha]
(z)$
is increasing on $z\in [0, \mathe^{-\alpha}]$ and $(iii)$ 
from (\ref{eq:maxhalpha}).

By \citet[Lemma 6]{berend2013sharp}, we have
$
\mexp{\ello{\empdist - \dist}} 
\le
\Lambda_n(\dist),
$
where 
$$\Lambda_n(\dist):=
2\sum_{\dist(j)<1/n} \dist(j) +
\frac{1}{\sqrt{n}}\sum_{\dist(j)\geq 1/n}\sqrt{\dist(j)}.$$
This quantity is always finite and
$
\Lambda_n(\dist)
\ninf
0
$
for all $\bmu\in\Delta_\N$
(ibid).
Thus, we obtain the bound
\beqn
\label{eq:expected_diff}
\mexp{\abs{\ent(\empdist) - \ent(\dist)}} 
&\leq& 
\left(4\alpha^\alpha +  2\ent[\alpha](\dist) \right)^{1/\alpha} \Lambda_n(\dist)^{1 - 1/\alpha}.
\eeqn

\paragraph*
{Finite support}
For distributions with a large support but concentrated mass, the bound in (\ref{eq:expected_diff}) 
compares favorably to
the state of the art, especially for smaller sample sizes. To illustrate this, consider a mixture of two
distributions with support sizes $d$ and $D$:
$\bmu'$ is uniform over $[d]$, $\bmu''$ is uniform
over $[d + D]$, 
and $\bmu:=p\bmu'+(1-p)\bmu''$,
for some $p\in[0,1]$.

The state-of-the-art upper bound for the plug-in estimator
can be inferred from
\citet[Appendix D]{wu2016minimax},
and has the form
\newcommand{\WY}{\operatorname{WY}}
\beq
\mexp{\abs{\ent(\empdist) - \ent(\dist)}}
&\le&
\WY(d,D,p,n) 
:=
\frac{d+D}{n}+\min \paren{C\frac{\log (d+D)}{\sqrt n}, \frac{\log n}{\sqrt n}}
\eeq
for some $C > 1$;
notice that it is insensitive to $p$.
For a fair comparison to (\ref{eq:expected_diff}),
our estimator's only a priori knowledge of $\bmu$
is that its support is of size at most $d+D$.
By Proposition~\ref{prop:max-alpha-ent},
we have
$
\max_{\bmu\in\Delta_K}
\ent[\alpha](\bmu)
\le 
\max\set{\alpha,\log K}^\alpha
+
(\alpha/\mathe)^\alpha
$.
This allows us to optimize over
$\alpha$ for each $n$:
\newcommand{\CKKW}{\operatorname{OUR}}
\beq
\CKKW(d,D,p,n) &:=& \inf_
{\alpha
>1
}
\Bigg( 
4
\alpha^\alpha + 
2
\max\set{\alpha,\log(d+D)}^\alpha
+
2(\alpha/\mathe)^\alpha
\Bigg)^{1/\alpha} \Lambda_n(\dist)^{1 - 1/\alpha}.
\eeq

\newcommand{\CT}{\operatorname{CT}}

Since $\bmu$ has finite support,
the Cover-Thomas inequality
(\ref{eq:CT}) also applies to
yield an adaptive estimate
when combined with (\ref{eq:ckw20}).
As $t\log(1/t)$ is concave, the latter
has the form
\beq
\mexp{\abs{\ent(\empdist) - \ent(\dist)}} 
&\leq&
\E{
\ello{\empdist-\dist}\log\frac{d+D}{\ello{\empdist-\dist}}
} \\
&\le&
\Lambda_n(\dist)\log\frac{d+D}{\Lambda_n(\dist)}
=:\CT(d,D,p,n)
.
\eeq
The comparisons are plotted in Figure~\ref{fig:CKKWvsWY} (Left).

\begin{figure}
\centering
\includegraphics[width=0.49\textwidth]{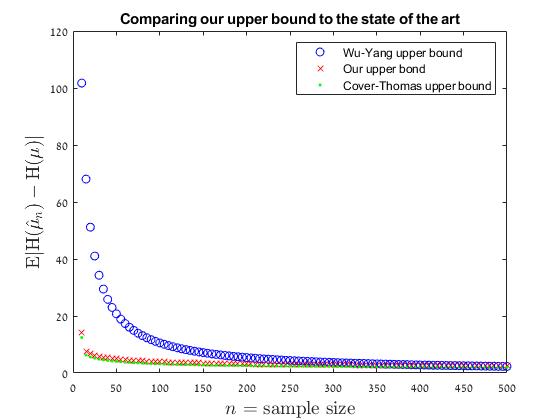}
\includegraphics[width=0.49\textwidth]{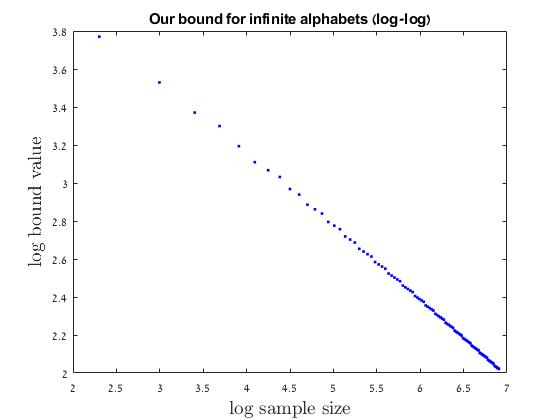}
\caption{
Left:
A 
comparison of the three bounds for
$d=10$,
$D=1000$,
$p=0.95$.
    Our bound considerably outperforms
    \citet{wu2016minimax}
    on small samples, and performs nearly as
    well as the finite-dimensional
    Cover-Thomas bound.
Right:
for our value of $q=2$, the log-log plot
shows roughly the correct slope of $-1/2$.
    }
    \label{fig:CKKWvsWY}
\end{figure}

\paragraph*{Infinite support}
In some cases our bound is nearly tight
(at least for the plug-in estimate), 
such as for
the family of zeta distributions $\bmu_q(i)\sim1/i^q$ with parameter $q 
>1
$. 
For this family,
\citet[Theorem~7]{antos2001convergence}
establish
a lower bound of
order
$
n^{\frac{1-q}{q}}
$
on $\mexp{\abs{\ent(\empdist) - \ent(\dist_q)}}$. 
It is straightforward to verify\footnote{
One can, for example, apply Cauchy's condensation
test,
followed up by the ratio test.
}
that
$\bmu_q\in\Delta_\N^{(\alpha)}$
for all $q,\alpha>1$.
Thus, we can optimize our bound
in (\ref{eq:expected_diff})
over all $\alpha>1$;
the results are presented in
Figure~\ref{fig:CKKWvsWY} (Right).

\section{Moments of Information}
\label{sec:moments}
We motivate our bounded-moment assumption
as
being 
\textit{considerably} less restrictive than the finite-alphabet assumptions and tail conditions studied in previous works
\citep{wu2016minimax,silva2018shannon};
see Section~\ref{sec:moments-tails} for a detailed
comparison.
Obtaining moment-based results
is essentially a desideratum laid out by \citet{antos2001convergence}, in which it is hypothesized that --- 
in parallel to the asymptotic distribution for the finite alphabet case --- 
moment conditions are the correct notion to achieve finite-sample estimates
in the infinite alphabet case. Our Theorem~\ref{thm:plug-in-minimax}(a) shows that, 
under these assumptions, 
there is an inverse logarithmic convergence rate (similar to, though distinct from, the results of \citet{wyner2003lower}) and, furthermore, using empirical quantities, this rate can be very much accelerated, as demonstrated in Corollary~\ref{cor:adaptive-hybrid-bound}.

In this section, we study some of the mathematical properties of moments of information.

\subsection{
Maximizing the $\alpha$-moment over a fixed support size
}
\begin{proposition}
\label{prop:max-alpha-ent}
For $K
\ge2
$ and $\alpha\ge1$,
\beq
\max\set{
\log K
,
(\alpha/\mathe)
}^\alpha
\;\le\;
\max_{\bmu\in\Delta_K}
\ent[\alpha](\bmu)
\;\le\;
\max\set{
\log K
,
\alpha
}^\alpha
+
(\alpha/\mathe)^\alpha
.
\eeq
\end{proposition}

We will need the following useful (and likely known) result.
\begin{lemma}[folklore]
\label{lem:folklore}
Suppose that
$0<a<1$
and
$f:[0,1]\to\R$
is 
strictly
concave on $[0,a]$
and 
strictly
convex on $[a,1]$.
Define the function $F:\Delta_K\to\R$
by
\beq
F(\bmu) = \sum_{i=1}^K f(\bmu(i)).
\eeq
Then any maximizer $\bmu^\star$
of $F$ is either the uniform distribution
or else has exactly $1$ ``heavy'' mass
$v\in[a,1]$
and $K-1$
identical
``light'' masses $(1-v)/(K-1)$.
\end{lemma}
\begin{proof}
A standard ``smoothing'' argument
\citep{loh13}
shows that if two 
masses $u\le v$ 
occur in the interval $(a,1)$,
there is an $\eps>0$
such that
$f(u-\eps)+f(v+\eps)> f(u)+f(v)$.
In other words, such masses can be
pushed apart 
(keeping their sum fixed)
to increase 
the value
of $F$, until one of them reaches
the boundary of $[a,1]$.
Furthermore, since $0<a<u<v$
and $u+v\le1$, repeated iteration
of the ``pushing apart'' operation
will hit the left endpoint (i.e., $a$)
rather than the right one (i.e., $1$).
Having exhausted the ``pushing apart'' process,
we are left with one ``heavy'' mass
$v\in[a,1]$
and $K-1$ ``lighter'' ones in $[0,a]$.
But concavity implies that $F$ will be maximized
by pulling the lighter masses in (as opposed to pushing them apart), which amounts to replacing
each of them by the average of the $K-1$ values.
\end{proof}

\begin{proof}[Proof of Proposition~\ref{prop:max-alpha-ent}]
Choosing $\bmu$ to be the uniform distribution
yields $\ent[\alpha](\bmu)=
\log^\alpha K
$,
and choosing $\bmu$
such that $v:=\bmu(1)=\mathe^{-\alpha}$
yields
$\ent[\alpha](\bmu)\ge
v\log(1/v)^\alpha=(\alpha/\mathe)^\alpha$.
Thus, the lower bound is proven
and
it only remains to prove the upper bound.

Let $\bmu^\star$ be a maximizer for given $\alpha, K$.
Recall the function 
$\loch[\alpha](z)=z \log^{\alpha}(1 /z)$
and note that it
is
strictly
concave on $[0,\mathe^{- (\alpha-1)}]$ and
strictly
convex on $[\mathe^{- (\alpha-1)},1]$.
Then Lemma~\ref{lem:folklore}
shows that $\bmu^\star$
will 
either be uniform or else
attains at most one value
$v\in [\mathe^{- (\alpha-1)},1]$
in the convex interval,
with the remaining values
equal to $\frac{1-v}{K-1}\in[0,\mathe^{- (\alpha-1)}]$
in the concave interval.
Only the latter case is non-trivial:
\beq
\ent[\alpha](\bmu^\star)
&=&
v\paren{\log\frac1v}^\alpha
+
(1-v)\paren{\log\frac{K-1}{1-v}}^\alpha
\eeq
for some $v$ satisfying
\beqn
\label{eq:Hv}
0<
\frac{1-v}{K-1}
\le
\mathe^{- (\alpha-1)}
\le v<1
.
\eeqn

Now
$v\paren{\log\frac1v}^\alpha$
is maximized 
over $[0,1]$
by 
$v=\mathe^{-\alpha}$,
which yields the value
$(\alpha/\mathe)^\alpha$.

To bound the second term,
$g(v):=(1-v)\paren{\log\frac{K-1}{1-v}}^\alpha$,
we consider two cases:
(i) $K-1<\mathe^\alpha$
and
(ii) $K-1\ge\mathe^\alpha$.
In case (i), $g$ is maximized
by $v^\star=1-(K-1)/\mathe^\alpha$
and
\beq
g(v^\star) =
(1-v^\star)\paren{\log\frac{K-1}{1-v^\star}}^\alpha
\le
\paren{\log\frac{K-1}{1-v^\star}}^\alpha
=\alpha^\alpha.
\eeq
In case (ii),
$g$ is monotonically decreasing in $v$.
The constraint
$\frac{1-v}{K-1}\le v$
from (\ref{eq:Hv})
implies 
$v\ge1/K$,
so in this case,
\beq
g(v)\le
\paren{\log\frac{K-1}{1-1/K}}^\alpha
=
\log^\alpha K.
\eeq

This proves the upper bound.

\end{proof}

\subsection{Moments of Information vs. Moments of Distributions}

Since for all $r\ge1$
and
$\bmu\in\Delta_\N$,
we trivially have
$\nrm{\bmu}_r\le 1$,
it is only for $r<1$
that $\nrm{\bmu}_r$ conveys
nontrivial tail information.
However, as a measure of tail decay,
the latter
is rather crude:
$\nrm{\bmu}_r<\infty$
for {\em any} $r<1$
implies
$\ent[\alpha](\bmu)<\infty$
for {\em all} $\alpha>0$:
\begin{proposition}
\label{prop:moment-vs-norm}
For all $\alpha>0$ and $r\in(0,1)$, we have
\beq
\ent[\alpha](\bmu)
&\le&
\paren{\frac{\alpha}{\mathe(1-r)}}^{2\alpha}
\nrm{\bmu}_r
.
\eeq
\end{proposition}
\noindent Remark.
The bound above is quite loose.
For example, for
$\alpha = 1$
the AM-GM inequality
readily yields,
for any $r \in (0,1)$,
\beq
\ent(\bmu) \leq 
\frac{r}{1-r}
\ln \nrm{\bmu}_{r}
.
\eeq
It may be of interest to investigate bounds
of the form
\beq
\ent[\alpha](\bmu) \leq 
a(\alpha) \log ^{b(\alpha)} \nrm{\bmu}_{c(\alpha)},
\eeq
for some functions $a,b,c:\alpha\mapsto(0,\infty)$.
\begin{proof}
We first claim that
for $\alpha>0$ and $r\in(0,1)$,
\beq
x\log^\alpha\frac{1}{x}<x^r,
\qquad
x\in[0,1]
.
\eeq
Indeed,
the function $x\mapsto x^{1-r}\log^\alpha(\frac{1}{x})$ is maximized at
$x=\mathe^{\frac{\alpha}{r-1}}$, attaining 
the
maximum value 
of
$\mathe^{-\alpha}(\frac{\alpha}{1-r})^\alpha$. The latter is less than $1$ whenever $\alpha<\mathe (1-r)$.
Likewise, whenever $\alpha<c\mathe (1-r)$, we have $x\log^\alpha(\frac{1}{x})<c^\alpha x^r$, and so $\ent[\alpha](\bmu)<c^{2\alpha} \nrm{\bmu}_r$.
Choosing $c=\frac{\alpha}{\mathe(1-r)}$
proves the claim.
\end{proof}

\subsection{Resolution of \citeauthor{jurgensen2010entropy} Conjectures}
\label{sec:conjectures}

In this section, we give a complete resolution
of the conjectures posed by
\citet{jurgensen2010entropy}.

\paragraph{Conjecture 10.1} 
\citet[Conjecture 10.1]{jurgensen2010entropy} posits that for $d=2$,
$\max_{\bmu\in\Delta_2}\ent[\alpha](\bmu)$ has two maximizers $\pi_1^{(\alpha)}=\left(\frac{1}{2}+x^{(\alpha)},\frac{1}{2}-x^{(\alpha)}\right)$
 and $\pi_2^{(\alpha)}=\left(\frac{1}{2}-x^{(\alpha)},\frac{1}{2}+x^{(\alpha)}\right)$ for some value $x^{(\alpha)}$ such that $x^{(2)} = \frac{1}{2e}\sqrt{e^2-4}$ and $x^{(\alpha)}$ is strictly increasing as $\alpha\to\infty$ and  
 $\lim_{\alpha\to\infty}x^{(\alpha)}=\frac{1}{2}$.
 
 By Lemma \ref{lem:folklore}, there are at most three maximizers. Since $\ent[\alpha]\left((\mathe^{-\alpha},1-\mathe^{-\alpha})\right)>(\frac{\alpha}{\mathe})^\alpha>\log^\alpha(2)$, the uniform distribution is not a maximizer. So, including permutations, there are exactly two maximizers.
 
 Let $(u_{\alpha}^\star,v_{\alpha}^\star)$ be the increasingly-ordered maximizing distribution.
 We cannot have
 $u_{\alpha}^\star<\mathe^{-\alpha}$, because this would only decrease $\ent[\alpha]$ as compared to $\ent[\alpha]\left((\mathe^{-\alpha},1-\mathe^{-\alpha})\right)$. By Lemma \ref{lem:folklore}, $u_{\alpha}^\star\leq \mathe^{-(\alpha-1)}$, and similarly 
$v_{\alpha}^\star\in [1-\mathe^{-(\alpha-1)},1-\mathe^{-\alpha}]$. By convexity and monotonicity of $\log \frac{1}{x}$ on $[0,1]$, the difference between $\abs{\loch[\alpha+1](\mathe^{-\alpha})-\loch[\alpha+1](u_{\alpha}^\star)}$ and
$\abs{\loch[\alpha](\mathe^{-\alpha})-\loch[\alpha](u_{\alpha}^\star)}$ shrinks by more than the difference between 
$\abs{\loch[\alpha+1](v_{\alpha}^\star)-\loch[\alpha+1](1-\mathe^{-\alpha})}$ and
$\abs{\loch[\alpha](v_{\alpha}^\star)-\loch[\alpha](1-\mathe^{-\alpha})}$ grows.  So, for $\max_{\bmu\in\Delta_2}\ent[\alpha](\bmu)$ to be less than $\max_{\bmu\in\Delta_2}\ent[\alpha+1](\bmu)$, as occurs (for sufficiently large $\alpha$) by resolution of Conjecture 10.4 below, it must be that $u_{\alpha+1}^\star<u_{\alpha}^\star$ and $v_{\alpha+1}^\star>v_{\alpha}^\star$ (as $\alpha$ tends to infinity).

Furthermore, 
$\mathe^{-(\alpha-1)}\to 0$
as $\alpha\to\infty$,
and so $\lim_{\alpha\to\infty}x^{(\alpha)}=\frac{1}{2}$.

To find the value of \(x^{(2)}\), set \(x \eqdef x^{(2)} \) and find the critical points of \[ \ent[2]\paren{x+\frac{1}{2}, x-\frac{1}{2}} 
\eqdef \left(\frac{1}{2}-x\right) \log ^2\left(\frac{1}{2}-x\right)+\left(x+\frac{1}{2}\right) \log ^2\left(x+\frac{1}{2}\right).\]
Differentiating and factoring, we get
\[ \frac{\mathd}{\mathd x} \ent[2]\paren{x+\frac{1}{2}, x-\frac{1}{2}} =
-\left(\log  \left(-x+\frac{1}{2}\right)-\log  \left(x+\frac{1}{2}\right)\right) \left(2+\log  \left(-x+\frac{1}{2}\right)+\log  \left(x+\frac{1}{2}\right)\right)=0. \]
Now \(x=0\) is a solution which we know is not the maximum and we also get \(x= \pm \frac{\sqrt{\mathe^2-4}}{2 \mathe}\) which exactly what \citet{jurgensen2010entropy} conjectured.

\paragraph{Conjecture 10.2}
\citet[Conjecture 10.2]{jurgensen2010entropy} posits that
for $\pi_1^{(\alpha)}, \pi_2^{(\alpha)}$ as above and
$\alpha\geq2$, 
we have
$\ent[\alpha](\pi_1^{(\alpha)})=\ent[\alpha](\pi_2^{(\alpha)})>(\log 2)^\alpha$ 
and moreover, this quantity
is strictly increasing and unbounded as $\alpha\to\infty$.

 Since $\log(2)<\frac{\alpha}{\mathe}$, by Proposition~\ref{prop:max-alpha-ent}, 
 $\ent[\alpha](\pi_1^{(\alpha)})=\ent[\alpha](\pi_2^{(\alpha)})>(\log 2)^\alpha$ and unbounded.

\paragraph{Conjecture 10.3}
\citet[Conjecture 10.3]{jurgensen2010entropy} posits that 
$\ent[\alpha]$ has $d$ local maxima
for
for $d>2$ and $\alpha>2$.
By 
Lemma~\ref{lem:folklore},
the only maxima are the uniform distribution and the $d$ permutations of 
$\sup_{v\in [\mathe^{-(\alpha-1)}, 1): (d-1)u+v=1 }(d-1)u\log^\alpha(u) +
v\log^\alpha(v)$, should the latter exist with $v$ in interior of interval. So there are either $1$ (e.g. $\mathe^\alpha=d$) or $d+1$ local maxima.

\paragraph{Conjecture 10.4}
For $d,\alpha\in\N$,
define $h^\star_{d,\alpha}
:=
\max_{\bmu\in\Delta_d}\ent[\alpha](\bmu)
$.
\citet[Conjecture 10.4]{jurgensen2010entropy}
posits that 
$
h^\star_{d,\alpha+1}
>
h^\star_{d,\alpha}
$
and that $
\lim_{\alpha\to\infty}
h^\star_{d,\alpha}
=\infty
$.
In light of the lower bound
in Proposition~\ref{prop:max-alpha-ent},
the latter claim (i.e., unboundedness)
is immediate.

For $d>\mathe^\alpha$, by Proposition \ref{prop:max-alpha-ent},
$\max_{\bmu\in\Delta_d}\ent[\alpha](\bmu)\leq 2\log^\alpha d$ and
$\log^{\alpha+1} d \leq\max_{\bmu\in\Delta_d}\ent[\alpha+1](\bmu)$.
We find, therefore, that for $d>\mathe^2$, 
$\max_{\bmu\in\Delta_d}\ent[\alpha](\bmu)\leq\max_{\bmu\in\Delta_d}\ent[\alpha+1](\bmu)$.

But since the conjecture takes interest in the case of $\alpha$ tending to infinity, let us focus on $\mathe^\alpha \geq d-1$.

By Lemma \ref{lem:folklore}, $\bmu_\alpha^\star:=\arg\max_{\bmu\in\Delta_d}\ent[\alpha](\bmu)$, is either uniform or takes two distinct values $v\in[\mathe^{-(\alpha-1)},1]$ and $u=\frac{1-v}{d-1}\in [0,\mathe^{-(\alpha-1)}]$.

For $x\in [0,\frac{1}{\mathe}]$, $\log(1/x)\geq 1$, so $\loch[\alpha+1](x)\geq\loch[\alpha](x)$.
So if $u,v \in [0,\frac{1}{\mathe}]$, then 
$\ent[\alpha](\bmu_\alpha^\star)< \ent[\alpha+1](\bmu_\alpha^\star) \leq \max_{\bmu\in\Delta_d}\ent[\alpha+1](\bmu)$

So assume instead that $v \in (\frac{1}{\mathe},1]$. In this case, we can bound the difference
$\loch[\alpha](v)-\loch[\alpha+1](v)\leq \loch[\alpha](v)\leq\frac{1}{\mathe}$.

Since $\mathe^\alpha \geq d-1$, $u>\mathe^{-\alpha}$ and must lie in $[\mathe^{-\alpha},\mathe^{-(\alpha-1)}]$. But for this entire interval, 
$x\in [\mathe^{-\alpha},\mathe^{-(\alpha-1)}]$ has 
$\loch[\alpha+1](x)- \loch[\alpha](x)\geq \frac{1}{\mathe} \geq \loch[\alpha](v)$.
In order to see this, it suffices, since $\loch[\alpha]$ and $\loch[\alpha+1]$ are both decreasing on $[\mathe^{-\alpha},\mathe^{-(\alpha-1)}]$, to show that 
$\loch[\alpha+1](\mathe^{-(\alpha-1)})-\loch[\alpha](\mathe^{-\alpha})> \frac{1}{\mathe}$. This can be seen by observing
$\mathe (\alpha-1)^{\alpha+1}-\alpha^\alpha>e^{\alpha-1}$ for all $\alpha\geq 3$.

It follows, therefore, that for $\alpha\in\N$, when $d>\mathe^\alpha$, for all $d\geq 8$, or for $\mathe^\alpha \geq d-1$, for all $\alpha\geq 3$, 
the former claim (monotonicity) holds, i.e.
$$\max_{\bmu\in\Delta_d}\ent[\alpha](\bmu) < \max_{\bmu\in\Delta_d}\ent[\alpha+1](\bmu).$$

(This can also be generalized to any $\beta(\neq\alpha+1)$ for sufficiently large $d$ or $\alpha$ respectively, if one so desired).

\paragraph{Our conjecture}
We close the section with a conjecture of our own.

\begin{conjecture}
For
$
\mathe^{\alpha}
<
d
<
\infty
$, 
we have
$\max_{\bmu \in \Delta_d} H_\alpha(\bmu) = \log^\alpha{d}$
and moreover,
the maximum is achieved by the uniform distribution over $[d]$.
\end{conjecture}

\section*{Acknowledgments}
This research was partially supported by
the Israel Science Foundation
(grant No. 1602/19) and Amazon Research Award.
GW is supported by the Special Postdoctoral Researcher Program (SPDR) of RIKEN.
We thank Ioannis Kontoyiannis, Igal Sason, and Sergio Verd\'u
for enlightening conversations.
We thank the anonymous referees for helpful suggestions and corrections.

%\bibliographystyle{plainnat}
%\bibliography{bibliography}

\appendix

\subsection{Comparison of tail-vs-moment assumptions}
\label{sec:moments-tails}
Expanding upon the observation of 
\citet{antos2001estimating}
that moment of information
assumptions
are
``weaker (and somewhat more natural''
than tail decay rates,
we make some concrete comparisons between the two.

Let us list a number of conditions one might impose
$\bmu$:
\begin{enumerate}
	\item[$A_1(\alpha)$:] Finite $\alpha$-moment of information \citet{antos2001estimating}: \\ For some $\alpha > 1, \E{\log^\alpha \frac{1}{\bmu(X)}} < \infty$.
	\item[$B_1(\alpha, M_\alpha)$:] Bounded $\alpha$-moment of information: \\ For some $\alpha > 1, \exists M_\alpha > 0,  \E{\log^\alpha \frac{1}{\bmu(X)}} < M_\alpha.$
	
	\item[$A_2(\beta)$:] Superlinear $\beta$ tail decay: \\ For some $\beta > 1, \bmu(i) = \bigO\left( \frac{1}{i^\beta} \right)$.
	\item[$B_2(\beta, \underline{c}_\beta, \overline{c}_\beta)$:] Controlled superlinear $\beta$ tail decay \citep{antos2001estimating}: \\ For some $\beta > 1, \exists \underline{c}_\beta, \overline{c}_\beta > 0$ such that $\frac{\underline{c}_\beta}{i^\beta} \leq \bmu(i) \leq \frac{\overline{c}_\beta}{i^\beta}$ .
	
	\item[$A_3(\gamma)$:] Superlinearly $\gamma$-logarithmic tail decay: \\ For some $\gamma > 1, \bmu(i) = \bigO\left( \frac{1}{i \log^\gamma{i}} \right)$.
	\item[$B_3(\gamma, \underline{c}_\gamma, \overline{c}_\gamma)$:] Controlled superlinearly $\gamma$-logarithmic tail decay  \citep{antos2001estimating}: 
	\\ For some $\gamma > 1, \exists \underline{c}_\gamma, \overline{c}_\gamma > 0$ such that $\frac{\underline{c}_\gamma}{i \log i^\gamma} \leq \bmu(i) \leq \frac{\overline{c}_\gamma}{i \log i^\gamma}$ .
	\end{enumerate}

\begin{proposition}
\begin{enumerate}
The following implications hold:
    \item[$(a)$] $A_3(\gamma), \gamma > 2 \implies A_1(\alpha), \forall \alpha < \gamma - 1$.
    \item[$(b)$] $A_1(\alpha), \alpha > 1 \centernot \implies A_3(\gamma), \gamma < \alpha + 1$
    \item[$(c)$] $B_1(\alpha, M_\alpha) \implies A_1(\alpha)$
    \item[$(d)$] $A_2(\beta) \implies A_3(\gamma), \forall \gamma > 1$
    \item[$(e)$] $A_2(\beta) \implies A_1(\alpha), \forall \alpha > 1$
    \item[$(f)$] $B_2(\beta, \underline{c}_\beta, \overline{c}_\beta) \implies A_2(\beta)$
    \item[$(g)$] $B_3(\gamma, \underline{c}_\gamma, \overline{c}_\gamma) \implies A_3(\gamma)$
    \item[$(h)$] $B_2(\beta, \underline{c}_\beta, \overline{c}_\beta) \implies B_1(\alpha, M_\alpha)$ with $M_\alpha = M_\alpha(\alpha, \overline{c}_\beta, \beta)$
    \item[$(i)$] $B_3(\gamma, \underline{c}_\gamma, \overline{c}_\gamma), \gamma > 2 \implies B_1(\alpha, M_\alpha), \forall \alpha < \gamma - 1$, with $M_\alpha = M_\alpha(\alpha, \overline{c}_\gamma, \gamma)$.
\end{enumerate}
\end{proposition}

\paragraph{Remark}
We start by noticing that
\begin{equation}
\label{equation:moment-of-information-governed-by-small-probabilities}
\begin{split}
\sum_{\substack{i \in \Omega \\ \bmu(i) \leq \mathe^{-\alpha}}} \bmu(i) \log^\alpha \frac{1}{\bmu(i)} \leq \sum_{i \in \Omega} \bmu(i) \log^\alpha \frac{1}{\bmu(i)} \leq \alpha^\alpha + \sum_{\substack{i \in \Omega \\ \bmu(i) \leq \mathe^{-\alpha}}} \bmu(i) \log^\alpha \frac{1}{\bmu(i)},
\end{split}
\end{equation}
and that on $[0, \mathe^{-\alpha}]$, it holds that $x \to x \log^\alpha{\frac{1}{x}}$ is increasing.
Since $\alpha^\alpha$ is finite, the convergence of the series is primarily governed by what happens or small probabilities.

\paragraph{Proof $(c), (d), (f), (g)$}
Immediate.

\paragraph{Proof $(a)$}
Suppose that assumption $A_3(\gamma)$ holds.
Then $\exists N \in \N, C > 0$ such that for any $i \geq N$, $\bmu(i) \leq C \frac{1}{i \log^\gamma{i}}$.
We focus on the rightmost term of \eqref{equation:moment-of-information-governed-by-small-probabilities}.
\begin{equation*}
\begin{split}
\sum_{\substack{i \in \Omega \\ \bmu(i) \leq\mathe^{-\alpha}}} \bmu(i) \log^\alpha \frac{1}{\bmu(i)} &\leq N\mathe^{-\alpha}\alpha^\alpha + 
\sum_{\substack{i \in \Omega \\ \bmu(i) \leq\mathe^{-\alpha} \\ i > N }} \bmu(i) \log^\alpha \frac{1}{\bmu(i)} \\
&\leq N\mathe^{-\alpha}\alpha^\alpha + 
C \sum_{\substack{i \in \Omega \\ \bmu(i) \leq\mathe^{-\alpha} \\ i > N }} \frac{1}{i \log^\gamma{i}} \log^\alpha \frac{i \log^\gamma{i}}{C} \\
&\leq N\mathe^{-\alpha}\alpha^\alpha + 
C \sum_{\substack{i \in \Omega \\ \bmu(i) \leq\mathe^{-\alpha} \\ i > N }} \frac{\left(\log i + \gamma\log \log i +  \log 1/C \right)^\alpha}{i \log^\gamma{i}}.
\end{split}
\end{equation*}
Since
$\log i$ dominates both $\log \log i$ and $\log 1/C$, 
the series converges whenever
$\sum_{i \in \N} \frac{1}{i \log^{\gamma - \alpha} i}$
converges, which occurs exactly for $\gamma > \alpha + 1$. 

\paragraph{Proof $(e)$}

Let $\bmu$, and suppose that $A_2(\beta)$ holds for some $\beta$.
Then $\exists N \in \N, C > 0$ 
such that $\forall i \in \N, i > N$, $\bmu(i) \leq C \frac{1}{i^\beta}$.
Let $\alpha > 1$. We decompose the expression of $\E{\log^\alpha\frac{1}{\bmu(X)}}$:
\begin{equation*}
\begin{split}
\sum_{i \in \Omega} \bmu(i) \log^\alpha \frac{1}{\bmu(i)} &= \sum_{\substack{i \in \Omega \\ i \leq N}} \bmu(i) \log^\alpha \frac{1}{\bmu(i)} + \sum_{\substack{i \in \Omega \\ i > N \\ \bmu(i) \leq\mathe^{-\alpha}}} \underbrace{\bmu(i) \log^\alpha \frac{1}{\bmu(i)}}_{x \mapsto x \log^\alpha \frac{1}{x} \text{ increasing on } [0,\mathe^{-\alpha}]} + \sum_{\substack{i \in \Omega \\ i > N \\ \bmu(i) >\mathe^{-\alpha}}} \bmu(i) \log^\alpha \frac{1}{\bmu(i)} \\
&\leq N \underbrace{\max_{x \in [0,1]} \set{x \log^\alpha \frac{1}{x}}}_{=\mathe^{-\alpha} \alpha^\alpha} + C \beta^\alpha \underbrace{\sum_{\substack{i \in \Omega \\ i > N\\ \bmu(i) \leq\mathe^{-\alpha}}} \frac{1}{i^\beta} \log^\alpha i}_{\leq S_{\alpha, \beta} < \infty} + \underbrace{\sum_{\substack{i \in \Omega \\ i > N \\ \bmu(i) >\mathe^{-\alpha}}} }_{\text{ at most $\mathe^\alpha$ elements}}\bmu(i) \log^\alpha \frac{1}{\bmu(i)}\\
&\leq (N\mathe^{-\alpha} + 1) \alpha^\alpha + C \beta^\alpha S_{\alpha, \beta},
\end{split}
\end{equation*}
such that for any $\alpha > 1$, there exists $M_\alpha < \infty$ 
that bounds the $\alpha$-moment of information. 
Notice that although existence is guaranteed, 
$N, C$ depend on the unknown $\bmu$.
The asymptotic nature of assumption $A_2(\beta)$
is therefore not enough to specify what $M_\alpha$ is.

\paragraph{Proof $(h)$}
Starting from \eqref{equation:moment-of-information-governed-by-small-probabilities},
\begin{equation*}
\begin{split}
\sum_{i \in \Omega} \bmu(i) \log^\alpha \frac{1}{\bmu(i)} &\leq \alpha^\alpha + \sum_{\substack{i \in \Omega \\ \bmu(i) \leq\mathe^{-\alpha}}} \bmu(i) \log^\alpha \frac{1}{\bmu(i)} \leq \alpha^\alpha + \sum_{\substack{i \in \Omega \\ \bmu(i) \leq\mathe^{-\alpha}}} \frac{\overline{c}_\beta}{i^\beta} \log^\alpha \frac{i^\beta}{\overline{c}_\beta} \\
\end{split}
\end{equation*}
which is upper bounded by a converging series, whose value is entirely defined by $\alpha, \beta$ and $\overline{c}_\beta$.

\paragraph{Proof $(i)$}
Follows the arguments of the proof for $(a)$ and the proof of $(h)$.
The series converges exactly when $\alpha < \gamma - 1$, and if it does, the value of the converging series is a function of $\alpha, \gamma, \overline{c}_\gamma$.

\paragraph{Proof $(b)$}
Assume that $\bmu$ satisfies $A_1(\alpha)$ for some $\alpha > 1$.
Then, $\forall \gamma < \alpha + 1$, the collection of distributions such that
$\bmu(i) \in \bigO \left( \frac{(\log{\log{i}})^\delta}{i \log^\gamma{i}} \right), \delta > 1$ also verifies the hypothesis.

\end{document}